\documentclass[11pt]{article}
\usepackage{amsthm,amssymb,amsfonts,amsmath,amscd}
\usepackage{nccmath}
\usepackage{graphicx}
\usepackage{hyperref}

\makeatletter \@addtoreset{equation}{section} \makeatother

\newtheorem{theorem}{Theorem}[section]
\newtheorem{lemma}{Lemma}[section]

\newtheorem{remark}{Remark}[section]

\newtheorem{proposition}{Proposition}[section]

\newcommand{\mdet}{\mathrm{det}}
\newcommand{\intd}{\displaystyle\int}
\newcommand{\Tr}{\mathrm{Tr}\,}
\newcommand{\Str}{\mathrm{Str}\,}

\begin{document}

\title{Universality for 1d random band matrices: sigma-model approximation}
\author{ Mariya Shcherbina
\thanks{Institute for Low Temperature Physics, Kharkiv, Ukraine\& Karazin Kharkiv National University, Kharkiv, Ukraine,  e-mail: shcherbi@ilt.kharkov.ua} \and
 Tatyana Shcherbina
\thanks{ Department of Mathematics, Princeton University, Princeton, USA, e-mail: tshcherbyna@princeton.edu. Supported in part by NSF grant DMS-1700009.}}

\date{}
\maketitle
\centerline{\textit{This paper is dedicated to Tom Spencer on the occasion of his 70th birthday.}}
\begin{abstract}
The paper continues the development of  the rigorous supersymmetric transfer matrix approach to the random band matrices started in \cite{SS:den}, \cite{SS:ChP}. We consider random Hermitian block band matrices consisting of $W\times W$ random Gaussian blocks (parametrized by  $j,k \in\Lambda=[1,n]^d\cap \mathbb{Z}^d$) with a
fixed entry's variance $J_{jk}=\delta_{j,k}W^{-1}+\beta\Delta_{j,k}W^{-2}$, $\beta>0$ in each block.  Taking the limit $W\to\infty$ with fixed $n$ and $\beta$, we
derive the sigma-model approximation of the second correlation function  similar to Efetov's one. 
Then, considering the limit $\beta, n\to\infty$, we prove that in the dimension $d=1$ the behaviour  of the sigma-model approximation in the bulk of the spectrum, as $\beta\gg n$,  is determined by the classical Wigner -- Dyson statistics. 
\end{abstract}

\section{Introduction}\label{s:1}
Random band matrices (RBM) represent quantum systems on a large box in $\mathbb{Z}^d$ with random
quantum transition amplitudes effective up to distances of order $W$, which is called a bandwidth.
They are natural intermediate models to study eigenvalue statistics
and quantum propagation in disordered systems as they interpolate between Wigner
matrices and random Schr$\ddot{\hbox{o}}$dinger operators: Wigner matrix ensembles represent mean-field
models without spatial structure, where the quantum transition rates between any
two sites are i.i.d. random variables; in contrast, random Schr$\ddot{\hbox{o}}$dinger operator has only  a random diagonal potential 
 in addition to the deterministic Laplacian on a box in $\mathbb{Z}^d$. 
 
 The density of states $\rho$ of a general class of RBM with $W\gg 1$ is given by the well-known Wigner semicircle law (see
\cite{BMP:91, MPK:92}):
\begin{equation}\label{rho_sc}
\rho(E)=(2\pi)^{-1}\sqrt{4-E^2},\quad E\in[-2,2].
\end{equation}
 The main feature of RBM is that they can be used
to model the celebrated Anderson metal-insulator phase transition in $d\ge 3$. Moreover,
the crossover for RBM can be investigated even in $d = 1$ by varying the bandwidth $W$.

More precisely, the key physical parameter of RBM  is the localization length $\ell_\psi$,
which describes the length scale of the eigenvector $\psi(E)$ corresponding to the energy $E\in (-2,2)$.
The system is called delocalized if for all $E$ in the bulk of spectrum $\ell_\psi$  is
comparable with the system size, $\ell_\psi\sim n$, and it is called localized otherwise.
Delocalized systems correspond to electric
conductors, and localized systems are insulators.
 
 In the case of 1d RBM there is a fundamental conjecture stating that for every  eigenfunction $\psi(E)$ in the bulk of the
 spectrum  $\ell_\psi$ is
of order $W^2$ (see \cite{Ca-Co:90, FM:91}). In $d=2$, the localization length is expected to be exponentially large in $W$, 
in $d\ge 3$ it is expected to be macroscopic, $\ell_\psi\sim n$, i.e. system is delocalized (for more details on these conjectures see \cite{Sp:12}).

The questions of the localization length  are closely related to the universality conjecture
of the bulk local regime of the random matrix theory. The bulk local regime deals with the behaviour of eigenvalues of $N\times N$
random matrices on the intervals whose length is of the order $O(N^{-1})$.
According to the Wigner -- Dyson universality conjecture, this local behaviour does not depend on the matrix probability
law (ensemble) and is determined only by the symmetry type of matrices (real
symmetric, Hermitian, or quaternion real in the case of real eigenvalues and orthogonal,
unitary or symplectic in the case of eigenvalues on the unit circle).
In terms of eigenvalue statistics
the conjecture about the localization length of RBM in $d=1$ means that  1d RBM in the bulk of the spectrum changes the spectral local behaviour of random operator type with
Poisson local eigenvalue statistics (for $W\ll \sqrt{N}$) to the local
spectral behaviour of the GUE/GOE type (for $W\gg \sqrt{N}$).

The conjecture supported by physical derivation due to Fyodorov and Mirlin (see \cite{FM:91}) based on supersymmetric formalism, 
and also by the so-called Thouless scaling. However, there are only a few partial results on the mathematical level of rigour.
At the present time only some upper and lower bounds for $\ell_\psi$ for the general class of 1d RBM are proved rigorously. 
It is known from the paper \cite{S:09} that $\ell_\psi\le W^8$. Recently this bound was improved in \cite{wegb:16} to $W^7$.
On the other side, for the general Wigner matrices (i.e. $W=n$) the
bulk universality has been proved in \cite{EYY:10, TV:11}, which gives $\ell_\psi \ge W$. 
By a development of the Erd\H{o}s-Yau approach,  there were also obtained some other results, where the localization length is controlled in a rather weak sense, i.e.
the estimates hold for \textit{``most''} eigenfunctions $\psi$ only: $\ell_\psi \ge W^{7/6}$ in \cite{EK:11} and $\ell_\psi\ge W^{5/4}$ in \cite{Yau:12}.
GUE/GOE gap distributions for $W\sim n$ was proved recently in \cite{BEYY:16}.

The study of the decay of eigenfunctions is closely related to properties of the Green function $(H-E-i\varepsilon)^{-1}$
with a small $\varepsilon$. For instance, if $(H-E-i\varepsilon)^{-1}_{ii}$ (without expectation) is bounded for all $i$
and some $E\in (-2,2)$,  then the normalized eigenvector $\psi(E)$ of $H$ is delocalized on scale $\varepsilon^{-1}$ in a sense that
\[
\max_i |\psi_i(E)|^2\lesssim \varepsilon,
\]
and so $\psi$ is supported on at least $\varepsilon^{-1}$ sites. In particular, if $(H-E-i\varepsilon)^{-1}_{ii}$ can be controlled down to
the scale $\varepsilon\sim 1/N$, then the system is in the complete delocalized regime.
Moreover, in view of the bound
\[
\mathbb{E}\{|(H-E-i\varepsilon)^{-1}_{jk}|^2\}\sim C\varepsilon^{-1}\, e^{-\|j-k\|/\ell}
\]
which is supposed to be valid for localized regime, the problem of localization/delocalization reduces to controlling
\[
\mathbb{E}\{|(H-E-i\varepsilon)^{-1}_{jk}|^2\}
\]
for $\varepsilon\sim 1/N$. As will be shown below, similar estimates  of $\mathbb{E}\{|\Tr (H-E-i\varepsilon)^{-1}|^2\}$ for $\varepsilon\sim N^{-1}$ 
are required to work with the correlation functions of RBM.

Despite many attempts, such control  has not been achieved so far. The standard approaches of \cite{EYY:10} and \cite{Yau:12} do 
not seem to work for $\varepsilon \le W^{-1}$, and so cannot give an information about the strong form of delocalization (i.e. for \textit{all} eigenfunctions). 
Classical moment methods, even with a delicate
renormalization approach \cite{S:11}, could not break the barrier $\varepsilon\sim W^{-1}$ either.

Another method, which allows to work with random operators with non-trivial spatial structures, is supersymmetry techniques (SUSY) based
on the representation of the determinant as an integral over the Grassmann variables.  
Combining this representation with the representation of the inverse determinant as an integral over the Gaussian complex field,
SUSY allows to obtain an integral representation for the main spectral characteristics  (such as density of states, second correlation functions, or
the average of an elements of the resolvent) as the averages of certain observables in some SUSY statistical mechanics models containing both
complex and Grassmann variables
(so-called \textit{dual} representation in terms of SUSY).
For instance, according to the properties of the Stieljes transform, the second correlation function $R_{2}$ defined by the equality
\begin{equation}\label{R_2} 
\mathbf{E}\Big\{ \sum_{j_{1}\neq  j_{2}}\varphi
(\lambda_{j_{1}},\lambda_{j_{2}})\Big\}=\int_{\mathbb{R}^{2}} \varphi
(\lambda_{1},\lambda_{2})R_{2}(\lambda_{1},\lambda_{2})
d\lambda_{1} d\lambda_{2},
\end{equation}
where $\{\lambda_j\}$ are eigenvalues of a random matrix, the function $\varphi: \mathbb{R}^{2}\rightarrow \mathbb{C}$ is
bounded, continuous and symmetric in its arguments, and the
summation is over all pairs of distinct integers $
j_{1},j_{2}\in\{1,\ldots,N\}$, can be rewritten as follows
\begin{align}\label{cor=det}
R_{2}(\lambda_1,\lambda_2)=&(\pi N)^{-2}\lim_{\varepsilon\to 0}
\mathbb{E}\{\Im\,\Tr(H-\lambda_1-i\varepsilon)^{-1}\Im\, \Tr (H-\lambda_2-i\varepsilon)^{-1}\}\\ \notag
=&(2i\pi N)^{-2}\lim_{\varepsilon\to 0}
\mathbb{E}\Big\{\Big(\Tr(H-\lambda_1-i\varepsilon)^{-1}-\Tr(H-\lambda_1+i\varepsilon)^{-1}\Big)\\ \notag
&\qquad\qquad\qquad\qquad\quad\quad\times \Big(\Tr(H-\lambda_2-i\varepsilon)^{-1}-\Tr(H-\lambda_2+i\varepsilon)^{-1}
\Big)\Big\},
\end{align}
and since
\begin{align}\label{r_det}
&\mathbb{E}\{\Tr (H-z_1)^{-1}\Tr (H-z_2)^{-1}\}= \frac{d^2}{dz_1'dz_2'}\mathbb{E}\Big\{\frac{ \det(H-z_1)\det(H-z_2))}{\det(H-z_1')\det(H-z_2'))}\Big\}\Big|_{z'=z},
\end{align}
$R_{2}$ can be represented as a sum of derivatives of the expectation of the ratio of four determinants, which we will call the
generalized correlation function. 

The derivation of SUSY integral representation is basically an algebraic step, and usually can be done by the standard algebraic manipulations.
SUSY is widely  used in
the physics literature, but the rigorous analysis of the obtained integral representation is a real mathematical challenge. Usually it is quite difficult, and it requires a powerful 
analytic and statistical mechanics techniques, such as a saddle point analysis, transfer operators, cluster expansions, renormalization
group methods, etc.
However, it can be done rigorously for some special class of RBM.  For instance, by using SUSY
the detailed information about the averaged density of states of a special case of Gaussian RBM in dimension 3 including local semicircle low
at arbitrary short scales  and smoothness in 
energy (in the limit of infinite volume and fixed large band width $W$) was obtained in \cite{DPS:02}. The techniques of that paper were used
in \cite{DL:16} to obtain the same result in 2d. A similar result in 1d was obtained by the SUSY transfer matrix approach in \cite{SS:den}. Moreover,
by applying the SUSY approach in  \cite{TSh:14}, \cite{SS:ChP} the crossover in this model (in 1d) was proved for the correlation functions
of characteristic polynomials. In addition, the rigorous application of SUSY to the Gaussian RBM which has the special block-band structure
was developed in \cite{TSh:14_1}, where the universality of the bulk 
local regime for $W\sim n$ was proved. The block band matrices are the special class of Wegner's orbital models (see \cite{We:79}), i.e.
Hermitian matrices $H_N$ with complex zero-mean random Gaussian entries $H_{jk,\alpha\beta}$,
where $j,k \in\Lambda=[1,n]^d\cap \mathbb{Z}^d$ (they parameterize the lattice sites) and $\alpha, \gamma= 1,\ldots, W$ (they
parameterize the orbitals on each site), such that
\begin{equation}\label{H}
\langle H_{j_1k_1,\alpha_1\gamma_1}H_{j_2k_2,\alpha_2\gamma_2}\rangle=\delta_{j_1k_2}\delta_{j_2k_1}
\delta_{\alpha_1\gamma_2}\delta_{\gamma_1\alpha_2} J_{j_1k_1}
\end{equation}
with
\begin{equation}\label{J_old}
J=1/W+\beta\Delta/W,
\end{equation}
where $W\gg 1$ and $\Delta$ is the discrete Laplacian on $\Lambda$. 
The probability law of $H_N$ can be written in the form
\begin{equation}\label{pr_l}
P_N(d H_N)=\exp\Big\{-\dfrac{1}{2}\sum\limits_{j,k\in\Lambda}\sum\limits_{\alpha,\gamma=1}^W
\dfrac{|H_{jk,\alpha\gamma}|^2}{J_{jk}}\Big\}dH_N.
\end{equation}
Combining the approach of \cite{TSh:14_1} with Green's function comparison strategy the delocalization (in a strong sense) for $W\gg n^{6/7}$  has been proved  in \cite{EB:15} for the block band matrices (\ref{H}) with rather general non-Gaussian element's distribution.

As it was mentioned above, the main advantage of SUSY techniques is that the main spectral
 characteristics of the model (\ref{H}) -- (\ref{J_old}) such as a density of states, $R_{2}$, $\mathbb{E}\{|G_{jk}(E+i\varepsilon)|^2\}$, etc. can
 be expressed via SUSY as the averages of certain observables in nearest-neighbour statistical mechanics models on $\Lambda$. This in particular
 in 1d case allows to combine the SUSY techniques with a transfer matrix approach. The supersymmetric transfer matrix formalism in this context 
 was first suggested by Efetov (see \cite{Ef}) and on a heuristic level it was adapted specifically for RBM
in \cite{FM:94} (see also references therein), although its rigorous application to the main spectral characteristics is quite difficult due to
the complicated structure and non self-adjointness of the corresponding transfer operator. The rigorous application of this method
to the density of states and correlation function of characteristic polynomials was done in \cite{SS:den}, \cite{SS:ChP}.
In this paper we make the next step in the developing of this approach and apply the technique to the so-called sigma-model approximation, 
which is often used by physicists to study complicated statistical mechanics systems.  In such approximation spins take values in
some symmetric space ($\pm 1$ for Ising model, $S^1$ for the rotator, $S^2$ for the classical
Heisenberg model, etc.). It is expected that sigma models have all the qualitative physics
of more complicated models with the same symmetry (for more detailes see, e.g., \cite{Sp:12}).  The sigma-model approximation for RBM was introduced by Efetov (see \cite{Ef}),
and the spins there are $4\times 4$ matrices with both complex and Grassmann entries (this approximation was studied in \cite{FM:91}, \cite{FM:94}).
 Let us mention also that the average conductance for 1d Efetov's sigma-model for RBM was
computed in \cite{SpDZ}. The aim of this paper is to derive the sigma-model approximation for the second correlation function for RBM and 
then analyse it rigorously in the dimension one by the transfer matrix formalism.

The mechanism of the crossover for the sigma-model is essentially the same as for the correlation functions of characteristic polynomials (see \cite{SS:ChP}). 
It is based on the fact that the spectral gap between two largest eigenvalues of the transfer operator  is $\beta^{-1}$ (it corresponds to $W^{-2}$ in \cite{SS:ChP}).
This implies that for $n/\beta\gg 1$ the $n$-th degree of the transfer operator converges to the rank one projection on the eigenvector corresponding to the largest 
eigenvalue, while for $n/\beta\ll 1$ the $n$-th degree of the transfer operator behaves like the multiplication operator. But  the structure of the transfer
operator for the sigma-model is more complicated: now it is a $6\times 6$ matrix kernel whose entries are kernels depending on two unitary $2\times 2$ matrices $U,U'$ and two hyperbolic  $2\times 2$ 
matrices $S, S'$. Hence the spectral analysis in the case of sigma-model is much more involved (see Section \ref{s:5}). We would like to mentioned that
in the case of the second generalized correlation function of the  1d block band matrices ((\ref{H})-(\ref{J_old}) with $\beta=\alpha W$), the transfer operator becomes
$70\times 70$ matrix, whose spectral analysis provides serious structural problems. Thus the analysis of the sigma-model approximation is an important intermediate step.

Set 
 \begin{align}\label{z}
 z_1&=E+i\varepsilon/N+\xi_1/N\rho(E),\quad z_2=
E+i\varepsilon/N+\xi_2/N\rho(E),\\ \notag
z_1^\prime&=E+i\varepsilon/N+\xi_1^\prime/N\rho(E),\quad z_2^\prime=
E+i\varepsilon/N+\xi_2^\prime/N\rho(E),
\end{align}
where $E\in (-2,2)$, $\varepsilon>0$, $\rho(E)$ is defined in (\ref{rho_sc}), and $\xi_1,\xi_2, \xi_1^\prime,
\xi_2^\prime\in [-C,C]\subset\mathbb{R}$ and define
\begin{align}\label{G_2}
\mathcal{R}_{Wn\beta}^{+-}(E,\varepsilon,\xi)&=\mathbf{E}\bigg\{\dfrac{\mdet(H_N-z_1)\mdet(H_N-\overline{z}_2)}
{\mdet(H_N-z_1')\mdet(H_N-\overline{z}_2^\prime)}\bigg\},\\ \notag
\mathcal{R}_{Wn\beta}^{++}(E,\varepsilon,\xi)&=\mathbf{E}\bigg\{\dfrac{\mdet(H_N-z_1)\mdet(H_N-z_2)}
{\mdet(H_N-z_1')\mdet(H_N-z_2^\prime)}\bigg\}
\end{align}
for $\xi=(\xi_1,\xi_2,\xi_1^\prime,\xi_2^\prime)$.

To derive the sigma-model approximation for the model (\ref{H}) -- (\ref{J_old}), we take $\beta$ in (\ref{J_old}) of order $1/W$, i.e. put
\begin{equation}\label{J}
J=1/W+\beta\Delta/W^2, \quad \beta>0.
\end{equation}
 The main result states that in the model (\ref{J}) with fixed $\beta$ and $|\Lambda|$, and with $W\to\infty$, the correlators
$\mathcal{R}_{Wn\beta}^{+-}$ and $\mathcal{R}_{Wn\beta}^{++}$ of (\ref{G_2}) converge to the values given by the sigma-model approximation. More precisely, we get
\begin{theorem}\label{thm:sigma_mod} 
Given $\mathcal{R}_{Wn\beta}^{+-}$ of (\ref{G_2}) ,(\ref{H}) and  (\ref{J}),
with any dimension $d$,  any fixed $\beta$, $|\Lambda|$, $\varepsilon>0$, and 
 $\xi=(\xi_1,\bar\xi_2, \xi_1',\bar\xi_2')\in\mathbb{C}^4$ 
 ($|\Im\xi_j|<\varepsilon\cdot \rho(E)/2$) we have,   as $W\to\infty$:
\begin{align}\label{sigma-mod}
&\mathcal{R}_{Wn\beta}^{+-}(E,\varepsilon, \xi)\to\mathcal{R}_{n\beta}^{+-}(E,\varepsilon, \xi),\quad
\frac{\partial^2\mathcal{R}_{Wn\beta}^{+-}}{\partial\xi_1'\partial\xi_2'}(E,\varepsilon,\xi)\to
\frac{\partial^2\mathcal{R}_{n\beta}^{+-}}{\partial\xi_1'\partial\xi_2'}(E,\varepsilon, \xi),
\\ \hbox{where}\quad&\mathcal{R}_{n\beta}^{+-}(E,\varepsilon, \xi)=
C_{E,\xi}\int \exp\Big\{\dfrac{\tilde\beta}{4}\sum\Str Q_jQ_{j-1}-\dfrac{c_0}{2|\Lambda|}\sum \Str Q_j\Lambda_{\xi, \varepsilon}\Big\} d Q,
\notag\end{align}
  $\tilde\beta=(2\pi\rho(E))^2\beta$, $U_j\in \mathring{U}(2)$, $S_j\in \mathring{U}(1,1)$ (see notation (\ref{U}) below),
\begin{equation*}
C_{E,\xi}=e^{E(\xi_1+\xi_2-\xi_1'-\xi_2')/2\rho(E)},
\end{equation*}
and $Q_j$ are $4\times 4$ supermatrices with commuting
diagonal and anticommution off-diagonal $2\times 2$ blocks
\begin{align}\label{Q}
Q_j=\left(\begin{array}{cc}
U_j^*&0\\
0&S_j^{-1}
\end{array}\right)\left(\begin{array}{cc}
(I+2\hat\rho_j\hat\tau_j)L&2\hat\tau_j\\
2\hat\rho_j&-(I-2\hat\rho_j\hat\tau_j)L
\end{array}\right)\left(\begin{array}{cc}
U_j&0\\
0&S_j
\end{array}\right),
\end{align}
\begin{align*}
d Q=\prod d Q_j,\quad d Q_j=(1-2n_{j,1}n_{j,2})\, d \rho_{j,1}d \tau_{j,1}\,d \rho_{j,2}d \tau_{j,2}\, d U_j\,d S_j
\end{align*}
with
\begin{align*}
&n_{j,1}=\rho_{j,1}\tau_{j,1},
\quad n_{j,2}=\rho_{j,2}\tau_{j,2},\\ \notag 
&\hat\rho_j=\mathrm{diag}\{\rho_{j1},\rho_{j2}\},\quad \hat\tau_j=\mathrm{diag}\{\tau_{j1},\rho_{j2}\},\quad L=\mathrm{diag}\{1,-1\}
\end{align*} 
Here $\rho_{j,l}$, $\tau_{j,l}$, $l=1,2$ are anticommuting Grassmann variables,
\[
\Str \left(\begin{array}{cc}
A&\sigma\\
\eta&B
\end{array}\right)=\Tr A-\Tr B,
\]
and 
\begin{align*}
\Lambda_{\xi,\varepsilon}=\mathrm{diag}\,\{\varepsilon-i\xi_1/\rho(E),-\varepsilon-i\xi_2/\rho(E),\varepsilon-i\xi_1'/\rho(E),-\varepsilon-i\xi_2'/\rho(E)\}.
\end{align*}\end{theorem}
\begin{theorem}\label{t:2}
Given $\mathcal{R}_{Wn\beta}^{++}$ of (\ref{G_2}) ,(\ref{H}) and  (\ref{J}),
with any dimension $d$,  any fixed $\beta$, $|\Lambda|$, $\varepsilon>0$, and 
 $\xi=(\xi_1,\xi_2, \xi_1',\xi_2')\in\mathbb{C}^4$ 
 ($|\Im\xi_j|<\varepsilon\cdot\rho(E)/2$) we have,   as $W\to\infty$:
\begin{align}\label{G++_lim}
&\mathcal{R}_{Wn\beta}^{++}(E,\varepsilon,\xi)\to e^{{ia_+}(\xi_1'+\xi_2'-\xi_1-\xi_2)/{\rho(E)}},\\
&\frac{\partial^2\mathcal{R}_{Wn\beta}^{++}}{\partial\xi_1'\partial\xi_2'}(E,\varepsilon,\xi)\to-a_+^2/\rho^2(E)\cdot e^{{ia_+}(\xi_1'+\xi_2'-\xi_1-\xi_2)/{\rho(E)}},
\qquad a_{+}=(iE+\sqrt{4-E^2})/{2}.
\notag\end{align}
\end{theorem}
Note that $Q_j^2=I$ for $Q_j$ of (\ref{Q}) and so the integral in the r.h.s of (\ref{sigma-mod}) is a sigma-model approximation similar to Efetov's one (see \cite{Ef}).

The next theorem describes the behaviour of $\mathcal{R}_{n\beta}^{+-}(E,\varepsilon,\xi)$ of the sigma-model (\ref{sigma-mod})
in the regime $n\to\infty $, $\beta> Cn \log^2 n$:
\begin{theorem}\label{t:1} If $n,\beta\to\infty$ in such a way that $\beta>Cn\log^2n$, then for any fixed $\varepsilon>0$ and 
 $\xi=(\xi_1,\xi_2, \xi_1',\xi_2')\in\mathbb{C}^4$ ($|\Im\xi_j|<\varepsilon\cdot\rho(E)/2$) we have
\begin{align}\label{t1.1}
\mathcal{R}^{+-}_{n\beta}\to &C_{E,\xi}\cdot e^{-c_0(\alpha_1+\alpha_2)}\Big(\delta_1\delta_2(e^{2c_0\alpha_1}-1)/\alpha_1\alpha_2
-(\delta_1+\delta_2)e^{2c_0\alpha_1}/\alpha_2+e^{2c_0\alpha_1}\alpha_1/\alpha_2\Big),\\
\label{alp}
\hbox{where}\quad&\alpha_1=\varepsilon-{i(\xi_1-\xi_2)}/{2\rho(E)},\quad \alpha_2=\varepsilon-{i(\xi_1'-\xi_2')}/{2\rho(E)},\\
&\delta_1={i(\xi_1'-\xi_1)}/{2\rho(E)},\quad\delta_2={i(\xi_2-\xi_2')}/{2\rho(E)}.
\notag\end{align}
\end{theorem}
Now Theorems \ref{thm:sigma_mod} -- \ref{t:1} and (\ref{cor=det}) -- (\ref{r_det}) imply the main result of the paper:
\begin{theorem}\label{thm:cor}
In the dimension $d=1$ the behavior of the sigma-model approximation of the second order correlation function (\ref{R_2}) of (\ref{H}), as $\beta\gg n$,
in the bulk of the spectrum coincides with those for the GUE. More precisely, if $\Lambda=[1,n]\cap \mathbb{Z}$ and $H_N$, $N=Wn$ are  matrices (\ref{H}) with $J$ of (\ref{J}), then 
for any $|E|<\sqrt{2}$
\begin{equation}\label{Un}
(N\rho(E))^{-2}
R_2\left(E+\displaystyle\frac{\xi_1}{\rho(E)\,N},
E+\displaystyle\frac{\xi_2}{\rho(E)\,N}\right)\longrightarrow
1-\dfrac{\sin^2 (\pi(\xi_1-\xi_2))}
{\pi^2(\xi_1-\xi_2)^2},
\end{equation}
in the limit first $W\to\infty$, and then $\beta, n\to\infty$, $\beta\ge Cn \log^2 n$.
\end{theorem}
\begin{remark} Notice that to prove universality of bulk local regime from the delocalization side of random block band matrices (\ref{H}) -- (\ref{pr_l}) without
a sigma-model approximation one have to take $J$ of (\ref{J_old}), fix $\beta$, and prove (\ref{Un}) in the limit $W,n\to\infty$, $W\gg n$, which is
different from the asymptotic regime considered in the current paper (first $W\to\infty$ with fixed $\beta$, then $\beta\gg n$, $\beta, n\to \infty$). 
\end{remark}
The paper is organized as follows. In Section \ref{s:2} we obtain a convenient  SUSY integral
representation for $\mathcal{R}_{Wn\beta}^{+-}$ and $\mathcal{R}_{Wn\beta}^{++}$ of (\ref{G_2}). In Section \ref{s:3} we prove 
Theorems \ref{thm:sigma_mod} and \ref{t:2}, in Section \ref{s:4} we  derive  Theorem \ref{thm:cor} from Theorems  \ref{t:2} and \ref{t:1}, 
in Section \ref{s:5} we prove Theorem \ref{t:1} modulo some auxiliary result proven in Section \ref{s:6}.

\bigskip 

\textbf{Acknowledgement.}
We are grateful to Yan Fyodorov for his suggestion of this particular
model for the derivation of  sigma-model approximation for RBM.
TS would like to thank Tom Spencer for his explanation of the nature of sigma-model approximation and for many
fruitful discussions without that this paper would never have been written. 

\subsection{Notation}
We denote by $C$, $C_1$, etc. various $|\Lambda|$, $\beta$, $W$-independent quantities below, which
can be different in different formulas. Integrals
without limits denote the integration (or the multiple integration) over the whole
real axis, or over the Grassmann variables.

Moreover,
\begin{itemize}

    \item[$\bullet$] $N=W|\Lambda|;$


    \item[$\bullet$] indices $i,j,k$ vary in $\Lambda$ and correspond to the number of the site
    (or the number of the block), index $l$ is always $1$ or $2$ (this is the field index),
    and Greek indices $\alpha, \gamma$ vary from $1$ to $W$ and correspond to the position of
    the element in the block;


    \item[$\bullet$] variables $\phi$ and $\Phi$ with different indices are complex variables or vectors
    correspondingly; if $x_j$ means some variable (vector or matrix) which corresponds to the site $j\in \Lambda$, then $x$ means
    vector $\{x_j\}_{j\in\Lambda}$,  $dx=\prod dx_j$, and $dx_j$ means the product of the differentials which correspond to functionally
    independent coefficients of $x_j$;

    \item[$\bullet$] variables $\psi$, $\Psi$, $\rho$, and $\tau$ with different indices are Grassmann variables or vectors or matreces
    correspondingly; if $\rho_j$ corresponds to the site $j\in \Lambda$, then $\rho$ means
    vector $\{\rho_j\}_{j\in\Lambda}$,  $d\rho=\prod d\rho_j$, and $d\rho_j$ means the product of the differentials which correspond the components
    (for vectors) or entries (for matrices)  taken into the lexicographic order; 
   
  \end{itemize}

\begin{fleqn}[18pt]
    \begin{align}\label{a_pm}
    \bullet \,\,\, & a_{\pm}=\dfrac{iE\pm\sqrt{4-E^2}}{2},\quad c_{\pm}=1+a_{\pm}^{-2},\quad  c_0=
\sqrt{4-E^2}=2\pi\rho(E);\\
    \label{L_pm}
    &L=\hbox{diag}\,\{1,-1\},\quad L_{\pm}=\hbox{diag}\,\{a_+,a_-\};
    \end{align}
  
  \begin{align}\label{U}
   \bullet \,\,\, &
  \mathring{U}(2)=U(2)/U(1)\times U(1),\quad \mathring{U}(1,1)=U(1,1)/U(1)\times U(1),
  \end{align}
   \hskip 1cm where $U(p)$ is a group of $p\times p$ unitary matrices, and $U(1,1)$ is a group of $2\times 2$ hyperbolic 
   
   \noindent \hskip 1cm matrices $S$ such that $S^*LS=L$;
 \end{fleqn}

\begin{fleqn}[18pt]

\begin{align}\label{L_cal}
\bullet \,\,\,\mathcal{L}_\pm(E)&=\Big\{r\Big(iE/2\pm\sqrt{4-E^2}/2\Big)|r\in [0,+\infty)\Big\};
\end{align}

\begin{align}
\bullet \,\,\,\tilde{\beta}=c_0^2\,\beta \label{beta_til};
\end{align}

\begin{align}\label{Z}
\bullet \,\,\,\ 
&Z_1=E\cdot I+i\varepsilon\cdot L/N+\hat{\xi}/N\rho(E), \quad\, Z_2=E\cdot I+i\varepsilon \cdot L/N+\hat{\xi}'/N\rho(E),\\
&Z_1^+=E\cdot I+i\varepsilon \cdot I/N+\hat{\xi}/N\rho(E), \quad Z_2^+=E\cdot I+i\varepsilon\cdot I/N+\hat{\xi}'/N\rho(E),\label{Z+}\\
&\hat\xi=\hbox{diag}\{\xi_1,\xi_2\},\quad  \hat\xi'=\hbox{diag}\{\xi_1',\xi_2'\}.\label{xi_hat}
\end{align}
\end{fleqn}

\section{Integral representations}\label{s:2}
%
In this section we perform the standard algebraic manipulations to obtain an integral representation for the determinant ratio $\mathcal{R}_{Wn\beta}^{+-}(E,\varepsilon,\xi)$ of (\ref{G_2}).
\begin{proposition}\label{p:int_repr}
For any dimension $d$, the determinant ratio $\mathcal{R}_{Wn\beta}^{+-}(E,\varepsilon,\xi)$ of (\ref{G_2}) can be written as follows:
\begin{align}\label{sup}
\mathcal{R}_{Wn\beta}^{+-}(E,\varepsilon,\xi)&=\dfrac{\mdet^2 J\cdot (-1)^{|\Lambda|W}}{(2\pi^3)^{|\Lambda|}\big((W-1)!(W-2)!\big)^{|\Lambda|}}\displaystyle\int 
dXdY \cdot\exp\big\{
i\sum\limits_{j\in \Lambda} \Tr Y_jLZ_2\big\}\\ \notag
 &\times\exp\Big\{-
\dfrac{1}{2}\sum\limits_{j,k\in\Lambda}J_{jk}\Tr (Y_jL)(Y_kL)-\dfrac{1}{2}\sum\limits_{j,k\in\Lambda}(J^{-1})_{jk}\Tr X_jX_k\Big\}\\ \notag
&\times \mdet
\big\{J^{-1}_{jk}\mathbf{1}_4-\delta_{jk}(iZ_1+X_j)^{-1}\otimes (Y_jL)\big\}_{j,k\in\Lambda}\prod\limits_{j\in\Lambda}\dfrac{\mdet^W (iZ_1+X_j)\mdet^W Y_j }{\mdet^2 Y_j},
\end{align}
where $\{X_j\}_{j\in\Lambda}$ are Hermitian  $2\times 2$ matrices with standard $dX_j$ , $\{Y_j\}_{j\in\Lambda}$ are $2\times 2$ 
positive Hermitian matrices with $dY_j$ of  Proposition \ref{p:supboz}, and $Z_{1,2}$ are defined in (\ref{Z}).

A similar formula is valid for $\mathcal{R}_{Wn\beta}^{++}(E,\varepsilon,\xi)$ with $Y_j$ instead of $Y_jL$ and $Z_l^+$ instead of $Z_l$, $l=1,2$
(see (\ref{Z+})).

\end{proposition}

\begin{proof}
Introduce complex and Grassmann fields:
\begin{align*}
\Phi_l=\{\phi_{jl}\}^t_{j\in\Lambda},&\quad \phi_{j l}=(\phi_{j l 1}, \phi_{j l 2},\ldots,
\phi_{j l W}),\quad l=1,2,
\quad -\quad \hbox{complex},\\
\Psi_l=\{\psi_{jl}\}^t_{j\in\Lambda}, &\quad \psi_{j l}=(\psi_{j l 1}, \psi_{j l 2},\ldots,
\psi_{j l W}),\quad l=1,2, \quad -\quad \hbox{Grassmann}.
\end{align*}
Using (\ref{G_C}) -- (\ref{G_Gr}) (see Appendix) we can write
\begin{equation*}
\begin{array}{c}
\mathcal{R}_{Wn\beta}^{+-}(E,\varepsilon,\xi)
=\pi^{-2W|\Lambda|}\mathbf{E}\Big\{\displaystyle\int \exp\{i\Psi_1^+(z_1^\prime-H_N)\Psi_1
-i\Psi_2^+(\overline{z}_2'-H_N)\Psi_2\}\\
\times\exp\{i\Phi_1^+(z_1-H_N)\Phi_1-i\Phi_2^+
(\overline{z}_2-H_N)\Phi_2\}d\Phi d\Psi\Big\}\\
=\displaystyle\int d\Phi d\Psi\,\, \exp\Big\{i(z_1^\prime\Psi_1^+\Psi_1
+z_1\Phi_1^+\Phi_1)-i(\overline{z}_2'\Psi_2^+\Psi_2
+\overline{z}_2\Phi_2^+\Phi_2)\Big\}\\
\times\mathbf{E}\Big\{\exp\Big\{-\sum\limits_{j\le k}\sum\limits_{\alpha, \gamma}
\Big(i\Re H_{jk,\alpha\gamma}\chi^+_{jk,\alpha\gamma}
-\Im H_{jk,\alpha\gamma}\chi^-_{jk,\alpha\gamma}\Big)\Big\}\Big\},
\end{array}
\end{equation*}
where $z_l, z_l^\prime$ are defined in (\ref{z}),
\begin{align*}
&\chi^{\pm}_{jk,\alpha\gamma}=\eta_{jk,\alpha\gamma}\pm \eta_{kj,\gamma\alpha},\\
&\eta_{jk,\alpha\gamma}=\overline{\psi}_{j 1\alpha}\psi_{k 1\gamma}-
\overline{\psi}_{j 2\alpha}\psi_{k 2\gamma}+\overline{\phi}_{j 1\alpha}\phi_{k 1\gamma}-
\overline{\phi}_{j 2\alpha}\phi_{k 2\gamma},\\
&\eta_{jj,\alpha\alpha}=(\overline{\psi}_{j 1\alpha}\psi_{j 1\alpha}-
\overline{\psi}_{j 2\alpha}\psi_{j 2\alpha}+\overline{\phi}_{j 1\alpha}\phi_{j 1\alpha}-
\overline{\phi}_{j 2\alpha}\phi_{j 2\alpha})/2.
\end{align*}
Averaging over (\ref{pr_l}), we get
\begin{align*}
\mathcal{R}_{Wn\beta}^{+-}(E,\varepsilon,\xi)=&\pi^{-2W|\Lambda|}\int d\Phi d\Psi\,\, \exp\Big\{i(z_1^\prime\Psi_1^+\Psi_1
+z_1\Phi_1^+\Phi_1)-i(\overline{z}_2'\Psi_2^+\Psi_2
+\overline{z}_2\Phi_2^+\Phi_2)\Big\}\\
&\times\exp\Big\{-\sum\limits_{j<k}\sum\limits_{\alpha,\gamma} J_{jk}\,\,
\eta_{jk,\alpha\gamma}\eta_{kj,\gamma\alpha}-\frac{1}{2}\sum\limits_{j, \alpha} J_{jj}\,\,
\eta_{jj,\alpha\alpha}^2\Big\}.
\end{align*}
Thus, 
\begin{equation}\label{G_av}
\begin{array}{c}
\mathcal{R}_{Wn\beta}^{+-}(E,\varepsilon,\xi)=\pi^{-2W|\Lambda|}\intd d\Phi d\Psi\,\, \exp\Big\{i\sum\limits_{j\in \Lambda} \Tr \tilde{X}_jLZ_1+
i\sum\limits_{j\in \Lambda} \Tr \tilde{Y}_jLZ_2\Big\}\\
\times \exp\Big\{\dfrac{1}{2}\sum\limits_{j,k\in\Lambda}J_{jk}\Tr (\tilde{X}_jL)(\tilde{X}_kL)-
\dfrac{1}{2}\sum\limits_{j,k\in\Lambda}J_{jk}\Tr (\tilde{Y}_jL)(\tilde{Y}_kL)\Big\}\\
\times \exp\Big\{-\sum\limits_{j,k\in\Lambda}J_{jk}
\big(\overline{\psi}_{j1}\psi_{k1}(\overline{\phi}_{k1}\phi_{j1}-
\overline{\phi}_{k2}\phi_{j2})+\overline{\psi}_{j2}\psi_{k2}(\overline{\phi}_{k2}\phi_{j2}-
\overline{\phi}_{k1}\phi_{j1})\big)\Big\},
\end{array}
\end{equation}
where $L$, $Z_{1,2}$ are defined in (\ref{L_pm}), (\ref{Z}), and
\begin{align*}
\tilde{X}_j=\left(
\begin{array}{ll}
\psi_{j1}^+\psi_{j1}& \psi_{j1}^+\psi_{j2}\\
\psi_{j2}^+\psi_{j1}& \psi_{j2}^+\psi_{j2}
\end{array}
\right),& \quad \tilde{Y}_j=\left(
\begin{array}{ll}
\phi_{j1}^+\phi_{j1}& \phi_{j1}^+\phi_{j2}\\
\phi_{j2}^+\phi_{j1}& \phi_{j2}^+\phi_{j2}
\end{array}
\right).
\end{align*}
Using the standard Hubbard-Stratonovich transformation, we obtain
\begin{multline}\label{Hub}
\big(2\pi^2\big)^{|\Lambda|}\mdet^{2}J\cdot\exp\Big\{\dfrac{1}{2}\sum\limits_{j,k\in\Lambda}J_{jk}\Tr (\tilde{X}_jL)(\tilde{X}_kL)\Big\}\\ 
=\int \exp\Big\{-\dfrac{1}{2}\sum\limits_{j,k\in\Lambda}(J^{-1})_{jk}\Tr X_jX_k+
\sum\limits_{j\in\Lambda}\Tr X_j\big(\tilde{X}_jL\big)\Big\}dX,
\end{multline}
where $X_j$ are $2\times 2$ Hermitian matrices with the standard measure $dX_j$. 

Substituting (\ref{Hub}) to (\ref{G_av}) and integrating over $d\Psi$ (see (\ref{G_Gr})), we get
\begin{eqnarray}\label{G_M}
\mathcal{R}_{Wn\beta}^{+-}(E,\varepsilon,\xi)&=\dfrac{\mdet^{-2}J}{\big(2\pi^{2(1+W)}\big)^{|\Lambda|}}
\int \exp\Big\{
i\sum_{j\in \Lambda} \Tr \tilde{Y}_jLZ_2-
\dfrac{1}{2}\sum_{j,k\in\Lambda}J_{jk}\Tr (\tilde{Y}_jL)(\tilde{Y}_kL)\Big\}\\ \notag
&\times \exp\Big\{-\dfrac{1}{2}\sum\limits_{j,k\in\Lambda}(J^{-1})_{jk}\Tr X_jX_k\Big\}\cdot \mdet\, M\cdot\,d\Phi\,
 \, dX
\end{eqnarray}
with $M=M^{(1)}-M^{(2)}$, where $M^{(1)}$ and $M^{(2)}$ are  $2W|\Lambda|\times 2W|\Lambda|$ matrices with entries
\begin{align}
\notag
M^{(1)}_{j\alpha l, k\gamma l'}&=\delta_{jk}\delta_{\alpha\gamma} (iZ_1+X_j)_{ll'}L_{ll},\quad
j,k\in \Lambda,\,\,\alpha,\gamma=1,\ldots, W,\,\,l,l'=1,2, \\ \label{M_1,2}
M^{(2)}_{j\alpha l, k\gamma l'}&=J_{jk}\delta_{ll'} L_{ll}\sum\limits_{\nu=1}^2\varphi_{j\alpha\nu}\overline{\varphi}_{k\gamma\nu} L_{\nu\nu}.
\end{align}
We can rewrite
\begin{equation}\label{mdet}
\mdet M=\mdet M^{(1)}\cdot \mdet \Big(1-\big(M^{(1)}\big)^{-1}M^{(2)}\Big)=:\mdet M^{(1)}\cdot \mdet \Big(1-\mathcal{M}\Big)
\end{equation}
with
\begin{equation}\label{Mcal}
\mathcal{M}_{j\alpha l, k\gamma l'}=J_{jk}(iZ_1+X_j)_{ll'}^{-1}\sum\limits_{\nu=1}^2\varphi_{j\alpha\nu}\overline{\varphi}_{k\gamma\nu} L_{\nu\nu}.
\end{equation}
Note that $\mathcal{M}=AB$, where
\begin{align}\notag
A_{j\alpha l, k\sigma l'}&=J_{jk}(iZ_1+X_j)^{-1}_{ll'}\,\varphi_{j\alpha\sigma}, \quad j,k\in \Lambda,\,\,\alpha,\gamma=1,\ldots, W,\,\,l,l',\sigma=1,2,\\ \label{AB}
B_{j\sigma l, k\alpha l'}&=\delta_{jk}\delta_{ll'}L_{\sigma\sigma}\,\overline{\varphi}_{k\alpha\sigma}. 
\end{align}
Therefore, using that $\mdet (1-AB)=\mdet (1-BA)$, (\ref{Mcal}), and (\ref{AB}), we get
\begin{equation}\label{ch_M}
\mdet \Big(1-\mathcal{M}\Big)=\mdet \Big(1-BA\Big)=:\mdet \Big(1-\tilde{\mathcal{M}}\Big),
\end{equation} 
where
\begin{align}\label{Mcal_til}
\tilde{\mathcal{M}}_{j\sigma l, k\sigma' l'}&=\sum\limits_{p,\alpha,\nu}B_{j\sigma l, p\alpha\nu} A_{p\alpha\nu, k\sigma' l'}=
J_{jk}(iZ_1+X_j)^{-1}_{ll'}\sum\limits_{\alpha=1}^W\overline{\varphi}_{j\alpha\sigma}\varphi_{j\alpha\sigma'}L_{\sigma\sigma}\\ \notag
&=J_{jk}(iZ_1+X_j)^{-1}_{ll'} (\tilde{Y}_jL)_{\sigma\sigma'}.
\end{align}
This yields
\begin{align}\label{det_fin}
\mdet \Big(1-\tilde{\mathcal{M}}\Big)&=\mdet \big\{\delta_{j,k}\mathbf{1}_4-J_{j,k}(iZ_1+X_j)^{-1}\otimes (\tilde{Y}_jL)\big\}_{j,k\in\Lambda}\\ \notag
&=\mdet^4 J\cdot \mdet\big\{J^{-1}_{jk}\mathbf{1}_4-\delta_{jk}(iZ_1+X_j)^{-1}\otimes (\tilde{Y}_jL)\big\}_{j,k\in\Lambda}.
\end{align}
Besides,
\begin{equation}\label{mdet1}
\mdet M^{(1)}=(-1)^{|\Lambda|W}\prod\limits_{j\in\Lambda}\mdet^W (iZ_1+X_j).
\end{equation}
Now substituting (\ref{M_1,2}) -- (\ref{Mcal}) and (\ref{ch_M}) -- (\ref{mdet1}) to (\ref{G_M}) and
applying the bosonization formula (see Proposition \ref{p:supboz}), we obtain (\ref{sup}).

The formula for $\mathcal{R}_{Wn\beta}^{++}(E,\varepsilon,\xi)$ can be obtained  by the same way.
\end{proof}


\section{Derivation of  the sigma-model approximation}\label{s:3}
\subsection{Proof of Theorem \ref{thm:sigma_mod}}

Let $\beta$ and $|\Lambda|$ be fixed, and  $W\to \infty$. 

Defining $|\Lambda|\times |\Lambda|$ matrix $R$ as
\begin{equation*}
J^{-1}=W\big(1-\dfrac{\beta}{W}\triangle+\dfrac{\beta^2}{W^2}\triangle^2-\ldots\big)=:W\big(1-\dfrac{\beta}{W}\triangle+\dfrac{1}{W^2}R\big),
\end{equation*}
putting $B_j=W^{-1}Y_jL$, and  shifting $iZ_1+X_j\to X_j$, we can rewrite (\ref{sup}) of Proposition \ref{p:int_repr} as
\begin{eqnarray} \notag
&\mathcal{R}_{Wn\beta}^{+-}(E,\varepsilon,\xi)=Q^{(1)}_{W, |\Lambda|}\displaystyle\int  dXdB\cdot\mdet 
\Big\{\big(\mathbf{1}-\dfrac{\beta}{W}\triangle+\dfrac{1}{W^2}R \big)_{jk}\mathbf{1}_4-\delta_{jk} \cdot X_j^{-1}\otimes B_j\Big\}_{j,k\in\Lambda}\\ \label{sup1}
&\times\exp\Big\{-\dfrac{W}{2}\sum\limits_{j\in \Lambda} \Big(\Tr (B_j-iZ_2)^2+\Tr (X_j-iZ_1)^2\Big)\Big\}\cdot\prod\limits_{j\in\Lambda}\dfrac{\mdet^W X_j\mdet^W B_j }{\mdet^2 B_j}\\ \notag
&\times\exp\Big\{
\dfrac{\beta}{2}\sum\limits_{j\sim k}\Big(\Tr (B_j-B_k)^2-\Tr (X_j-X_k)^2\Big)+\dfrac{1}{2W}\sum\limits_{j, k}
R_{jk}\,\Tr (X_j-iZ_1)(X_k-iZ_1)\Big\},
\end{eqnarray}
where 
\begin{align*}
Q^{(1)}_{W, |\Lambda|}&=\dfrac{\mdet^2 J\cdot W^{2(W+1)|\Lambda|}\cdot e^{-W|\Lambda|\Tr Z_2^2/2}}{(2\pi^3)^{|\Lambda|}\big((W-1)!(W-2)!\big)^{|\Lambda|}}\\ \notag
&=\dfrac{ W^{4|\Lambda|}\cdot e^{2W|\Lambda|-W|\Lambda|\Tr Z_2^2/2}}{(2\pi^2)^{2|\Lambda|}}\cdot \Big(1+O\big(W^{-1}\big)\Big).
\end{align*}
Change the variables to
\begin{align*}
X_j&=U_j^*\hat{X}_jU_j,\,\,\,\,\,\, \hat{X}_j=\hbox{diag}\,\{x_{j,1},x_{j,2}\},\,\,\,\,  U_j\in \mathring{U}(2),
\quad\,\,\, x_{j,1}, x_{j,2}\in \mathbb{R},\\ \notag
B_j&=S_j^{-1}\hat{B}_jS_j,\,\,\,\, \hat{B}_j=\hbox{diag}\,\{b_{j,1},b_{j,2}\},\quad
S_j\in \mathring{U}(1,1),
\,\,\, b_{j,1}\in \mathbb{R}^+,\, b_{j,2}\in \mathbb{R}^-.
\end{align*}
The Jacobian of such a change is
\[
2^{|\Lambda|}(\pi/2)^{2|\Lambda|}\prod\limits_{j\in\Lambda}(x_{j,1}-x_{j,2})^2
\prod\limits_{j\in\Lambda}(b_{j,1}-b_{j,2})^2.
\]
This and (\ref{sup1}) yield
\begin{align}\label{G_last}
&\mathcal{R}_{Wn\beta}^{+-}(E,\varepsilon,\xi)=Q^{(2)}_{W, |\Lambda|}
\int  dSdU
\int d x\int_{\mathbb{R}_+^{|\Lambda|}\times \mathbb{R}_-^{|\Lambda|}} d b
\cdot \prod\limits_{j\in\Lambda}\dfrac{(x_{j,1}-x_{j,2})^2
(b_{j,1}-b_{j,2})^2}{b_{j,1}^2b_{j,2}^2}\\ \notag
&\times \mdet \,\mathcal{D}(\hat{X}, \hat{B}, U, S)\cdot
\exp\Big\{-W\sum\limits_{j\in\Lambda}\sum\limits_{l=1}^2\left(f(x_{j,l})+f(b_{j,l})\right)\Big\}\\ \notag
&\times \exp\Big\{
\dfrac{\beta}{2}\sum\limits_{j\sim k}\Big(\Tr (S_j^{-1}\hat{B}_jS_j-S_k^{-1}\hat{B}_kS_k)^2-\Tr (U_j^*\hat{X}_jU_j-U_k^*\hat{X}_kU_k)^2\Big)\Big\}\\ \notag
&\times\exp\Big\{\dfrac{1}{2W}\sum\limits_{j, k}
R_{jk}\,\Tr (U_j^*\hat{X}_jU_j-iZ_1)(U_k^*\hat{X}_kU_k-iZ_1)\Big\}\\ \notag
&\times\exp\Big\{\dfrac{i}{|\Lambda|}\sum\limits_{j\in\Lambda}\Big(\Tr U_j^*\hat{X}_jU_j\big(i\varepsilon L+\hat{\xi}/\rho(E)\big)+\Tr S_j^{-1}\hat{B}_jS_j\big(i\varepsilon L+\hat{\xi}'/\rho(E)\big)\Big)\Big\},
\end{align}
where 
\begin{eqnarray}\label{D}
\det \mathcal{D}(\hat{X}, \hat{B}, U, S)=\det\Big\{\Big(\mathbf{1}-\dfrac{\beta}{W}\triangle+\dfrac{1}{W^2}R \Big)_{jk}\mathbf{1}_4-\delta_{jk} \cdot X_j^{-1}\otimes B_j\Big\}_{j,k\in\Lambda}\\ \notag
=\det
\Big\{\delta_{jk}\big(\mathbf{1}-\hat{X}_j^{-1}\otimes \hat{B}_j\big)+\dfrac{1}{W}\Big(-\beta\triangle+\dfrac{1}{W}R \Big)_{jk}\cdot U_jU_k^*\otimes S_jS_k^{-1}\Big\}_{j,k\in\Lambda},
\end{eqnarray}
\begin{align*}
Q^{(2)}_{W, |\Lambda|}&=2^{|\Lambda|}(\pi/2)^{2|\Lambda|}\cdot e^{W|\Lambda|(\Tr Z_1^2+\Tr Z_2^2)/2-W|\Lambda| (2+E^2)} \cdot Q^{(1)}_{W, |\Lambda|}\\ \notag
&=\dfrac{ W^{4|\Lambda|}\cdot e^{E(\xi_1+\xi_2)/\rho(E)}}{2^{3|\Lambda|}\pi^{2|\Lambda|}}\cdot \Big(1+O\big(W^{-1}\big)\Big),\\ 
\notag
f(x)&=x^2/2-iE x-\log x -(2+E^2)/4.
\end{align*}
The constant in $f(x)$ is chosen in such a way  that $\Re f(a_\pm)=0$. Measures $dU_j$, $dS_j$ in (\ref{G_last}) are the Haar measures over $\mathring{U}(2)$ and $\mathring{U}(1,1)$ correspondingly.

Also it is easy to see that for $|E|\le\sqrt{2}$ we can deform the contours of integration as 
\begin{itemize}
\item for $x_{j,1}$, $x_{j,2}$ to $iE/2+\mathbb{R}$;
\item for $b_{j,1}$ to $\mathcal{L}_+(E)$ of (\ref{L_cal});
\item for $b_{j,2}$ to $\mathcal{L}_-(E)$ of (\ref{L_cal}).
\end{itemize}

To prove Theorem \ref{thm:sigma_mod}, we are going to integrate (\ref{G_last}) over the ``fast'' variables: $\{x_{j,l}\}, \{b_{j,l}\}$, $l=1,2$, $j\in \Lambda$.
The first step is the following lemma:
\begin{lemma}\label{l:s_point}
The integral (\ref{G_last}) over $\{x_{j,l}\}, \{b_{j,l}\}$, $l=1,2$, $j\in \Lambda$ can be restricted to the integral over the $W^{-(1-\kappa)/2}$-neighbourhoods (with a  small $\kappa>0$) of the points
\begin{itemize}
\item[I.] $x_{j,1}=a_+$, $x_{j,2}=a_-$ or $x_{j,1}=a_-$, $x_{j,2}=a_+$, $b_{j,1}=a_+$, $b_{j,2}=a_-$ for any $j\in\Lambda$;
\item[II.] $x_{j,1}=x_{j,2}=a_+$, $b_{j,1}=a_+$, $b_{j,2}=a_-$ for any $j\in\Lambda$;
\item[III.] $x_{j,1}=x_{j,2}=a_-$, $b_{j,1}=a_+$, $b_{j,2}=a_-$ for any $j\in\Lambda$.
\end{itemize}
Moreover, the contributions of the points II and III are $o(1)$, as $W\to\infty$.
\end{lemma}
\begin{proof}
The proof of the first part of the lemma is straightforward and  based on the fact that $\Re f(z)$ for $z=x+iE/2$, $x\in \mathbb{R}$ has 
two global minimums at $z=a_\pm$, and for $z\in\mathcal{L}_\pm(E)$ has one global minimum at $z=a_\pm$.

To prove the second part of the lemma, consider the neighbourhood of the point II (the point III can be treated in a similar way). 
Change the variables as
\begin{align}\label{change1}
\begin{array}{ll}
x_{j,1}=a_++{\tilde{x}_{j,1}}/{\sqrt{W}}, &x_{j,2}=a_++\tilde{x}_{j,2}/{\sqrt{W}},\\ 
b_{j,1}=a_+\big(1+{\tilde{b}_{j,1}}/{\sqrt{W}}\big), &b_{j,2}=a_-\big(1+{\tilde{b}_{j,2}}/{\sqrt{W}}\big).
\end{array}
\end{align}
This gives the Jacobian $(-1)^{|\Lambda|} W^{-2|\Lambda|}$ and also the additional $W^{-|\Lambda|}$ since
\begin{align*}
x_{j,1} -x_{j,2}=(\tilde{x}_{j,1}-\tilde{x}_{j,2})/\sqrt{W}.
\end{align*}
Together with $Q^{(2)}_{W, |\Lambda|}$ this gives $W^{|\Lambda|}$ in front of the integral (\ref{G_last}).
In addition, expanding $f$ into the series, we get
\begin{align}\label{f_exp}
&f(x_{j,l})=f(a_+)+\dfrac{c_+}{2}\frac{\tilde{x}_{j,l}^2}{W}-\dfrac{1}{2a_+^3}\frac{\tilde{x}_{j,l}^3}{W^{3/2}}+O\Big(\dfrac{\tilde{x}_{j,l}^4}{W^{2}}\Big), \quad l=1,2\\ \notag
&f(b_{j,1})=f(a_+)+\dfrac{a_+^2c_+}{2}\cdot \dfrac{\tilde{b}_{j,1}^2}{W}-\dfrac{1}{2}\cdot \dfrac{\tilde{b}_{j,1}^3}{W^{3/2}}+O\Big(\dfrac{\tilde{b}_{j,1}^4}{W^{2}}\Big),\\ \notag
& f(b_{j,2})=f(a_-)+\dfrac{a_-^2c_-}{2}
\cdot \dfrac{\tilde{b}_{j,2}^2}{W}-\dfrac{1}{2}\cdot \dfrac{\tilde{b}_{j,2}^3}{W^{3/2}}+O\Big(\dfrac{\tilde{b}_{j,2}^4}{W^{2}}\Big),
\end{align}
where
\begin{equation}\label{c_pm}
c_\pm=1+a_\pm^{-2}, \quad f(a_+)=-f(a_-)\in i\mathbb{R}.
\end{equation}

We are going to compute the leading order of the integral over $\{\tilde{x}_{j,l}\}, \{\tilde{b}_{j,l}\}$, $l=1,2$, $j\in \Lambda$. To this end,  we leave 
the quadratic part of $f$ (see (\ref{f_exp})) in the exponent,  expand everything else into the series of $\tilde{x}_{j,l}/\sqrt{W}, \tilde{b}_{j,l}/\sqrt{W}$ 
around the saddle-point $\tilde{x}_{j,l}=\tilde{b}_{j,l}=0$, and compute the Gaussian integral of each term of this expansion.
We are going to prove that all this terms are $o(1)$.

Indeed, consider the expansion of the diagonal elements of $\mathcal{D}(\hat{X}, \hat{B}, U, S)$ of (\ref{D}):
\begin{align}\notag
&d_{j,l1}=1-x_{j,l}^{-1}b_{j,1}=(\tilde{x}_{j,l}/a_+-\tilde{b}_{j,1})/\sqrt{W}+({\tilde{x}_{j,l}\tilde{b}_{j,1}/a_+-\tilde{x}_{j,l}^2/a_+^2}){W}+O\big(W^{-3(1-\kappa)/2}\big),\\ 
&d_{j,l2}=1-x_{j,l}^{-1}b_{j,2}=c_--({\tilde{x}_{j,l}/a_+-\tilde{b}_{j,2}})/a_-^2\sqrt{W}+O\big(W^{-1+\kappa}\big),\quad l=1,2.\label{d_exp1_1}
\end{align}
If we rewrite the determinant of $\mathcal{D}(\hat{X}, \hat{B}, U, S)$  in a standard way, then each summand has strictly one element from each row and column.
Because of (\ref{d_exp1_1}), each element in the rows $(j,11)$ and $(j,21)$ has at least $W^{-1/2}$, and so the expansion of $\mdet \,\mathcal{D}(\hat{X}, \hat{B}, U, S)$ starts from $W^{-|\Lambda|}$. Moreover, to obtain $W^{-|\Lambda|}$ (i.e. non-zero contribution) we must consider the summands of the determinant expansion that have only diagonal elements $d_{j,ls}$ (since non-diagonal elements of $\mathcal{D}(\hat{X}, \hat{B}, U, S)$ are $O(W^{-1})$ or less), and furthermore only the first terms in the expansions (\ref{d_exp1_1}) and all other function
in (\ref{G_last}). Thus we get
\begin{align}\label{++}
C\cdot \Big\langle\prod\limits_{j\in\Lambda} \dfrac{\tilde{x}_{j,1}/a_+-\tilde{b}_{j,1}}{\sqrt{W}}\cdot \dfrac{\tilde{x}_{j,2}/a_+-\tilde{b}_{j,1}}{\sqrt{W}} \cdot (\tilde{x}_{j,1}
-\tilde{x}_{j,2})^2\Big\rangle_{++}+o(1),
\end{align}
where
\begin{align*}
\Big\langle \cdot \Big\rangle_{++}=\int \Big(\cdot\Big) \exp\Big\{-\frac{1}{2}\sum_{j\in\Lambda}\Big({c_+(\tilde x_{j,1}^2+\tilde x_{j,2}^2)}+
{a_+^2c_+ \tilde b_{j,1}^2}
+{a_-^2c_- \tilde b_{j,2}^2}\Big)\Big\}  d\tilde{x}\, d\tilde{b}.
\end{align*}
But it is easy to see that  the Gaussian integral in (\ref{++}) is zero, which completes the proof of the lemma.
\end{proof}

According to Lemma \ref{l:s_point} the main contribution to (\ref{G_last}) is given by the neighbourhoods of the saddle points $x_{j,1}=a_+$, $x_{j,2}=a_-$ or $x_{j,1}=a_-$, $x_{j,2}=a_+$. All such points can be obtained from each other by rotations of $U_j$, 
so we can consider only $x_{j,1}=a_+$, $x_{j,2}=a_-$ for all $j\in\Lambda$. Similarly to the proof of Lemma \ref{l:s_point}, 
change variables as
\begin{align}\label{change}
\begin{array}{ll}
x_{j,1}=a_++{\tilde{x}_{j,1}}/{\sqrt{W}}, &x_{j,2}=a_-+{\tilde{x}_{j,2}}/{\sqrt{W}},\\ 
b_{j,1}=a_+\big(1+{\tilde{b}_{j,1}}/{\sqrt{W}}\big), &b_{j,2}=a_-\big(1+{\tilde{b}_{j,2}}/{\sqrt{W}}\big).
\end{array}
\end{align}
That slightly change the expansions (\ref{f_exp})  and (\ref{d_exp1_1}). We get
\begin{align}\label{f_exp2}
&f(x_{j,2})=f(a_-)+\dfrac{c_-}{2}\cdot \dfrac{\tilde{x}_{j,2}^2}{W}-\dfrac{1}{2a_-^3}\cdot \dfrac{\tilde{x}_{j,2}^3}{W^{3/2}}+O\Big(\dfrac{\tilde{x}_{j,2}^4}{W^{2}}\Big),
\end{align}
\begin{align}\label{d_exp1}
&d_{j,11}=1-x_{j,1}^{-1}b_{j,1}=\dfrac{\tilde{x}_{j,1}/a_+-\tilde{b}_{j,1}}{\sqrt{W}}+
\dfrac{a_+\tilde{x}_{j,1}\tilde{b}_{j,1}-\tilde{x}_{j,1}^2}{a_+^2W}+O\big(W^{-3(1-\kappa)/2}\big),\\ \notag
&d_{j,22}=1-x_{j,2}^{-1}b_{j,2}=\dfrac{\tilde{x}_{j,2}/a_--\tilde{b}_{j,2}}{\sqrt{W}}+\dfrac{a_-\tilde{x}_{j,2}\tilde{b}_{j,2}-\tilde{x}_{j,2}^2}
{a_-^2W}+O\big(W^{-3(1-\kappa)/2}\big),\\
\notag
&d_{j,12}=1-x_{j,1}^{-1}b_{j,2}=c_+-\dfrac{\tilde{x}_{j,1}/a_+-\tilde{b}_{j,2}}{a_+^{2}\sqrt{W}}-\dfrac{a_+\tilde{x}_{j,1}\tilde{b}_{j,2}-\tilde{x}_{j,1}^2}{a_+^4W}
+O\big(W^{-3(1-\kappa)/2}\big), \\ \notag
&d_{j,21}=1-x_{j,2}^{-1}b_{j,1}=c_- -\dfrac{{\tilde{x}_{j,2}/a_--\tilde{b}_{j,1}}}{a_-^{2}\sqrt{W}}-\dfrac{a_-\tilde{x}_{j,2}\tilde{b}_{j,1}-\tilde{x}_{j,2}^2}{a_-^4W}
+O\big(W^{-3(1-\kappa)/2}\big).
\end{align}
The change (\ref{change}) gives the Jacobian $W^{-2|\Lambda|}$, which together with $Q^{(2)}_{W, |\Lambda|}$ gives $W^{2|\Lambda|}$ in front of the integral (\ref{G_last}).
Similarly to the proof of Lemma \ref{l:s_point}, we are going to compute the leading order of the integral (\ref{G_last}) over $\{\tilde{x}_{j,l}\}, \{\tilde{b}_{j,l}\}$, $l=1,2$, $j\in \Lambda$, and so  we leave 
the quadratic part of $f$ (see (\ref{f_exp}) and (\ref{f_exp2})) in the exponent,  expand everything else into the series of $\tilde{x}_{j,l}/\sqrt{W}, \tilde{b}_{j,l}/\sqrt{W}$ 
around the saddle-point $\tilde{x}_{j,l}=\tilde{b}_{j,l}=0$, and compute the Gaussian integral of each term of this expansion. We are going to prove, that the non-zero
contribution is given by the terms having at least $W^{-2|\Lambda|}$.

\begin{lemma}\label{l:det_exp} Formula (\ref{G_last}) can be rewritten as
\begin{align}\label{G_main}
\mathcal{R}_{Wn\beta}^{+-}&(E,\varepsilon,\xi)=(c_0/2\pi)^{2|\Lambda|}C_{E,\varepsilon}\int  dz\, d\tilde\rho\, d\tilde\tau \,d U\, d S\,\\ &\times\exp\Big\{-\dfrac{1}{2}(Mz, z)+W^{1/2}(z,h^0)+
W^{-1/2}(z,h+\zeta/|\Lambda)|)\Big\} \notag \\ \notag
&\times \exp\Big\{{\beta}\sum\Tr\Big(U_j^*\tilde\rho_jS_j-U_{j-1}^*\tilde\rho_{j-1}S_{j-1}\Big) \Big(S_j^{-1}\tilde\tau_jU_j-
S_{j-1}^{-1}\tilde\tau_{j-1}U_{j-1}\Big)\Big\}\\ \notag
&\times \exp\Big\{\sum\Big(c_+ n_{j,12}+c_- n_{j,21}- n_{j,1}/{c_0a_+}+n_{j,2}/{c_0a_-}\Big)-\beta c_0^2\sum (v_j^2+t_j^2)\Big\}\\ \notag
&\times\exp\Big\{\dfrac{ic_0}{2|\Lambda|}\sum\limits_{j\in\Lambda}\Big(\Tr U_j^*L U_j\big(i\varepsilon L+\hat{\xi}/\rho(E)\big)+\Tr S_j^{-1} L S_j\big(i\varepsilon L+
\hat{\xi}'/\rho(E)\big)\Big)\Big\}
+o(1),
\end{align}
where
\begin{align}
\label{rt_tilde}
&\tilde\rho_j=\left(\begin{array}{cc}
\rho_{j,11}&\rho_{j,12}/\sqrt{W}\\
\rho_{j,21}/\sqrt{W}&\rho_{j,22}
\end{array}\right), \quad \tilde\tau_j=\left(\begin{array}{cc}
\tau_{j,11}&\tau_{j,12}/\sqrt{W}\\
\tau_{j,21}/\sqrt{W}&\tau_{j,22}
\end{array}\right)
\\
\notag
&n_{j,12}=\rho_{j,12}\tau_{j,12},\quad n_{j,21}=\rho_{j,21}\tau_{j,21},\\ \notag
&n_{j,1}=\rho_{j,11}\tau_{j,11},\quad\,\, n_{j,2}=\rho_{j,22}\tau_{j,22},\\ \notag
&z=(z_{j,11},z_{j,22},z_{j,12},z_{j,21})=(\tilde x_{j,1},\tilde x_{j,2},\tilde b_{j,1},\tilde b_{j,1}),
\end{align}
and 
\begin{align}\label{M}
&M=M_0+W^{-1}\tilde{M}\\
\label{M_0}
&(M_0 z,z)=\sum\limits_{j\in \Lambda}\Big(c_+\tilde x_{j,1}^2+c_-\tilde x_{j,2}^2+a_+^2c_+\tilde b_{j,1}^2+a_-^2c_-\tilde b_{j,2}^2\Big) \\
\label{tilde_M}
&(\tilde{M}z,z)=-2\beta\sum \Big(\tilde{x}_{j,1}\tilde{x}_{j-1,1}+\tilde{x}_{j,2}\tilde{x}_{j-1,2}-a_+^2\tilde{b}_{j,1}\tilde{b}_{j-1,1}-a_-^2\tilde{b}_{j,2}\tilde{b}_{j-1,2}\Big)\\ \notag
&+2 \beta \sum \Big( v_j^2\,(\tilde x_{j,1}-\tilde x_{j,2})(\tilde x_{j-1,1}-\tilde x_{j-1,2})+t_j^2\,(a_+\tilde b_{j,1}-a_-\tilde b_{j,2})(a_+\tilde b_{j-1,1}-a_-\tilde b_{j-1,2})\Big)\\ \notag
&-\sum\Big(\dfrac{4}{c_0^2} (\tilde x_{j,1}\tilde x_{j,2}-\tilde b_{j,1}\tilde b_{j,2})
-2(a_+^{-3}n_{j,12}\tilde x_{j,1}\tilde b_{j,2}+a_-^{-3}n_{j,21}\tilde x_{j,2}\tilde b_{j,1})\Big).
\end{align}
Here $\zeta=\{\zeta_j\}_{j\in\Lambda}$, $\zeta_j=(\zeta_{j,11},\zeta_{j,22}, a_+\zeta_{j,12},a_-\zeta_{j,21})$ with
\begin{align*}
&\zeta_{j,11}={-\varepsilon+i\xi_1/\rho(E)}
+2\alpha_1u_j^2,\quad
\zeta_{j,22}={\varepsilon+i\xi_2/\rho(E)}
-2\alpha_1u_j^2,\\ \notag
&\zeta_{j,12}={-\varepsilon+i\xi_1'/\rho(E)}
-2\alpha_2s_j^2,\quad
\zeta_{j,21}={\varepsilon+i\xi_2'/\rho(E)}
+2\alpha_2s_j^2,
\end{align*}
where $\alpha_{1,2}$ are defined in (\ref{alp}).
We also denoted
\begin{equation}\label{h}
\begin{array}{ll}
h=\{h_{j,ls}\}_{j\in\Lambda,l,s=1,2},& h^0=\{h^0_{j,ls}\}_{j\in\Lambda,l,s=1,2},\\ 
h_{j,11}={2}/{c_0}-\beta c_0 v_j^2-\beta c_0 v_{j+1}^2+{a_-n_{j,12}}/{a_+^2}, \quad &h^0_{j,11}={n_{j,1}}/{a_+},\\ 
h_{j,22}=-{2}/{c_0}+\beta c_0 v_j^2+\beta c_0 v_{j+1}^2+{a_+n_{j,21}}/{a_-^2}, \quad &h^0_{j,22}={n_{j,2}}/{a_-},\\ 
h_{j,12}={2a_+}/{c_0}-2-\beta c_0 a_+ t_j^2-\beta c_0 a_+ t_{j+1}^2-n_{j,21}{a_+}/{a_-},\quad  &h^0_{j12}=-{n_{j,1}},\\ 
h_{j,21}=-{2a_-}/{c_0}-2+\beta c_0 a_- t_j^2+\beta c_0 a_- t_{j+1}^2-n_{j,12}{a_-}/{a_+},  \quad &h^0_{j,21}=-{n_{j,2}},\\
\end{array}
\end{equation}
and
\begin{align*}
u_j=|(U_j)_{12}|,\quad  v_j=|(U_jU_{j-1}^*)_{12}|,\quad s_j=|(S_j)_{12}|,\quad t_j=|(S_jS_{j-1}^{-1})_{12}|.
\end{align*}
\end{lemma}
\begin{proof} Rewriting the determinant in (\ref{D}) in a standard way, we obtain
\begin{equation}\label{exp_det}
\mdet\,\mathcal{D}(\hat{X}, \hat{B}, U, S)= \sum\limits_{\bar\sigma} (-1)^{|\sigma|} \prod\limits_{j\in |\Lambda|} P_{j,\bar\sigma_j}(\tilde{x}_{j,1},\tilde{x}_{j2}, \tilde{b}_{j,1},\tilde{b}_{j,1}),
\end{equation}
where $\bar\sigma$ is a permutation of $\{(j,ls)\}$, $l,s=1,2$, $j\in \Lambda$, $\bar\sigma_j$ is its restriction on   $\{(j,ls)\}_{l,s=1}^2$, $(-1)^{|\sigma|}$ is a sign of
$\sigma$ and $P_{j,\bar\sigma_j}$ is an expansion in $\tilde{x}_{j,1},\tilde{x}_{j2}$, $\tilde{b}_{j,1},\tilde{b}_{j,1}$ of the product of four elements from the rows 
$\{(j,ls)\}_{l,s=1}^2$ taken with respect to $\bar\sigma_j$.
 
 Let us prove that
for each $j\in\Lambda$ and any $\bar\sigma$ each term of $P_{j,\bar\sigma_j}(\tilde{x}_{j,1},\tilde{x}_{j2}, \tilde{b}_{j,1},\tilde{b}_{j,1})$ of (\ref{exp_det}) belongs to
one of the three following groups:
\begin{itemize}
\item[i.] has a coefficient $W^{-2}$ or lower;

\item[ii.] has a coefficient $W^{-3/2}$ and at least one of variables $\tilde{x}_{j,1},\tilde{x}_{j2}$, $\tilde{b}_{j,1},\tilde{b}_{j,1}$ of the odd degree;

\item[iii.] has a coefficient $W^{-1}$ and at least two variables of $\tilde{x}_{j,1},\tilde{x}_{j2}$, $\tilde{b}_{j,1},\tilde{b}_{j,1}$ of the odd degree;
\end{itemize}

Note that each element in the expansion of the coefficients of the rows $(j,11)$ and $(j,22)$ has a coefficient $W^{-1/2}$ or lower,
and so $P_{j,\bar\sigma_j}(\tilde{x}_{j,1},\tilde{x}_{j2}, \tilde{b}_{j,1},\tilde{b}_{j,1})$ has a coefficient $W^{-1}$ or lower. 
In addition, if $P_{j,\bar\sigma_j}(\tilde{x}_{j,1},\tilde{x}_{j,2}, \tilde{b}_{j,1},\tilde{b}_{j,1})$ contains any terms with $R_{jk}$ (see (\ref{D})),
or at least one off-diagonal elements in $(j,12)$ and $(j,21)$, we get a coefficient  $W^{-2}$ or lower (and so obtain the group (i)).

We are left to consider terms with $d_{j,12}d_{j,21}$. If   
$P_{j,\bar\sigma_j}(\tilde{x}_{j,1},\tilde{x}_{j,2}, \tilde{b}_{j,1},\tilde{b}_{j,1})$ contains two off-diagonal elements in rows $(j,11)$ and $(j,11)$,
we get group (i). One off-diagonal element and $d_{j,11}$ (or $d_{j,22}$) gives group (ii) or group (i) (since off-diagonal elements do not
depend on $\tilde{x}_{j,1},\tilde{x}_{j,2}$, $\tilde{b}_{j,1},\tilde{b}_{j,1}$), and it is easy to see from (\ref{d_exp1}) that all the terms in expansion of
$d_{j,11}d_{j,22}d_{j,12}d_{j,21}$ belongs to groups (i) -- (iii).

To get a non-zero contribution, we have to complete  the expression $P_{j,\bar\sigma_j}(\tilde{x}_{j,1},\tilde{x}_{j,2}, \tilde{b}_{j,1},\tilde{b}_{j,1})$ by some other terms of 
the expansion of the exponent of (\ref{G_last}) in order to get an even degree of each variable $\tilde{x}_{j,1},\tilde{x}_{j,2}$, $\tilde{b}_{j,1},\tilde{b}_{j,1}$. 
But all such a terms have the coefficient $W^{-1/2}$ or lower, and therefore Lemma \ref{l:det_exp} yields that the coefficient near each $j$ in terms that
gives a non-zero contribution must be $W^{-2}$ or lower. Since we have a coefficient $W^{2|\Lambda|}$ in (\ref{G_last}) after the change (\ref{change}), 
this means that to get a non-zero
contribution each coefficient must be exactly $W^{-2}$. Note that the terms of $P_{j,\bar\sigma_j}(\tilde{x}_{j,1},\tilde{x}_{j,2}, \tilde{b}_{j,1},\tilde{b}_{j,1})$ 
that can be completed to the monomial with all even degrees and with a coefficients $W^{-2}$ does not contain any terms with $R_{jk}$,
and any terms of the expansion $d_{j,ls}$, $l,s=1,2$ of order $W^{-3/2}$ or lower. They also cannot be completed to the monomial with all even degrees 
and with a coefficients $W^{-2}$ by any terms of the exponent of (\ref{G_last}) that has a coefficient lower then $W^{-1/2}$ for some $j$.
Thus we need to consider the terms up to the third order in the expansions (\ref{f_exp}) and (\ref{f_exp2}), the linear terms of the functions in the second and the forth
exponents of (\ref{G_last}), and the linear terms in $-2\log  b_{j,l}$, $l=1,2$ coming from
\[
 b_{j,l}^{-2}=e^{-2\log b_{j,l}},\quad l=1,2.
\] 
Note that the terms containing $\tilde x_{j,1}\tilde b_{j,1}/W$ in $d_{j,11}$ (see (\ref{d_exp1})) cannot contribute to the limit, since if we complete them
to the monomial with even degrees of $\tilde x_{j,1},\tilde b_{j,1}$, then it will contain $W^{-2}$ and an additional $W^{-1}$ should come from the line
containing $d_{j,22}$. Moreover, the terms containing $\tilde x_{j,1}^2$ in $d_{j,11}$ can give a non-zero contribution only if the resulting monomial
contains only $\tilde x_{j,1}^2$, since otherwise, taking into account the contribution of the line
containing $d_{j,22}$, we again obtain at least $W^{-3}$. Thus we can replace $\tilde x_{j,1}^2$ by its average via Gaussian measure 
$(2\pi/c_+)^{-1/2} e^{-c_+\tilde x_{j,1}^2/2}$, i.e. by $ c_+^{-1}$.
The same is true for $\tilde x_{j,2}\tilde b_{j,2}/W$ and for $\tilde x_{j,2}^2$ which could be replaced by $ c_-^{-1}$. Similar argument yields that
the contribution of the terms with $\tilde x_{j,1}^2$ in the line containing $d_{j,12}$ and  $\tilde x_{j,2}^2$ in the line containing $d_{j,21}$  disappear in the limit
$W\to\infty$. Thus the term corresponding to
$W^{2|\Lambda|}\det\mathcal{D}$ in (\ref{G_last})
can be replaced by the term
\begin{align}\label{det'}
\int d\rho\,d\tau \exp&\Big\{\beta\sum\Tr\Big(U_j^*\tilde\rho_jS_j-U_{j-1}^*\tilde\rho_{j-1}S_{j-1}\Big) \Big(S_j^{-1}\tilde\tau_jU_j-
S_{j-1}^{-1}\tilde\tau_{j-1}U_{j-1}\Big)\\ \notag
&+\sum\Big(c_+n_{j,12}+c_-n_{j,21}-n_{j,1}/c_0a_++ n_{j,2}/c_0a_- \Big)\\ \notag
&+{W}^{1/2}\sum\Big(\big(\tilde{x}_{j,1}/a_+-\tilde{b}_{j,1})n_{j,1}+(\tilde{x}_{j,2}/a_--\tilde{b}_{j,2})n_{j,2}\big)\\
&  -W^{-1/2}\sum\Big(a_+^{-2}\big(\tilde{x}_{j,1}/a_+-\tilde{b}_{j,2}\big)n_{j,12}+a_-^{-2}\big(\tilde{x}_{j,2}/a_--\tilde{b}_{j,1}\big)n_{j,21}\Big)\Big\}
+O(W^{-1/2}),
\notag\end{align}
where $\tilde \rho_j$, $\tilde \tau_j$, $n_{j,12}$, $n_{j,21}$, $n_{j,1}$, $n_{j,2}$ are defined in (\ref{rt_tilde}).
Here we have used Grassmann variables $\{\rho_{j,ls}\}$, $\{\tau_{j,ls}\}$, $j\in \Lambda$, $l,s=1,2$ to rewrite the determinant (\ref{D}) with respect to (\ref{G_Gr}),
have substituted (\ref{d_exp1}) and left only terms that give the contribution (according to arguments above), and then have changed $\rho_{j,11}\to \sqrt{W}\rho_{j,11}$,  
$\tau_{j,11}\to \sqrt{W}\rho_{j,11}$.
Note also 
\begin{equation}\label{capm}
c_+a_+^2=c_0a_+,\quad c_-a_-^2=-c_0a_-.
\end{equation}

Now let us prove that the contribution of the third order in the expansions (\ref{f_exp}) and (\ref{f_exp2}) is small. 
Indeed, the terms  $P_{j,\bar\sigma_j}(\tilde{x}_{j,1},\tilde{x}_{j,2}, \tilde{b}_{j,1},\tilde{b}_{j,1})$ 
that can be completed to the monomial with all even degrees and with a coefficients $W^{-2}$ by these cubic terms can be one of two types

1. terms $\big(\tilde x_{j,1}/a_+-\tilde{b}_{j,1}\big)\cdot x\cdot c_+\cdot c_-$, where $c_+$, $c_-$ come from the zero terms of $d_{j,12}$, $d_{j,21}$ 
(see (\ref{d_exp1})) and $x$ is an element of the row $(j,22)$ and so
does not depend on $\tilde{x}_{j,1}$, $\tilde{b}_{j,1}$ (or similar terms with $\big(\tilde x_{j,2}/a_--\tilde{b}_{j,2}\big)$);

2. terms of $\big(\tilde x_{j,1}/a_+-\tilde{b}_{j,1}\big)\big(\tilde x_{j,2}/a_--\tilde{b}_{j,2}\big) \big(\tilde x_{j,1}/a_+-\tilde{b}_{j,2}\big) \cdot c_-$
with $\tilde x_{j,1}^2$ or $\tilde b_{j,2}^2$ (or similar terms with $c_+$ coming  from $d_{j,12}$) ;

But it is easy to see that
\[
\int \big(\tilde x^4_{j,1}/(3a_+^4)-\tilde{b}_{j,1}^4/3\big) \cdot e^{-\frac{c_+\tilde x^2_{j,1}}{2}-\frac{a_+^2c_+\tilde b^2_{j,1}}{2}}\,d\tilde x_{j,1}\, d\tilde b_{j,1}= \dfrac{2\pi}{a_+c_+}\Big(\dfrac{1}{a_+^4c_+^2}-\dfrac{1}{a_+^4c_+^2}\Big)=0,
\]
and so the contribution of (1) is zero. Similarly the contribution (2) is zero.

Therefore, the contribution of the third order in the expansions (\ref{f_exp})  is small,  and  
using (\ref{det'}) and also
\begin{align*}
&\exp\Big\{\dfrac{i}{|\Lambda|}\sum\limits_{j\in\Lambda}\Big(\Tr U_j^*L_\pm U_j\big(i\varepsilon L+\hat{\xi}_1/\rho(E)\big)+\Tr S_j^{-1} L_\pm S_j\big(i\varepsilon L+
\hat{\xi}_2/\rho(E)\big)\Big)\Big\}\\
=&\exp\big\{-E(\xi_1+\xi_2+\xi_1'+\xi_2')/2\rho(E)\big\}\\
&\times \exp\Big\{\dfrac{ic_0}{2|\Lambda|}\sum\limits_{j\in\Lambda}\Big(\Tr U_j^*L U_j\Big(i\varepsilon L+
\frac{\hat{\xi}}{\rho(E)}\Big)+\Tr S_j^{-1} 
LS_j\Big(i\varepsilon L+
\frac{\hat{\xi}'}{\rho(E)}\Big)\Big)\Big\}
\end{align*}
for $L_\pm$, $L$ defined in (\ref{L_pm}), we get (\ref{G_main}). 

\end{proof}
Denoting the exponent in the second line of (\ref{G_main}) by $\mathcal{E}(z)$ and taking the Gaussian integral over $d z$ with $z$ of (\ref{rt_tilde}), we get
\begin{align}\label{G_main1}
\int_{\mathbb{R}^{4|\Lambda|}} \mathcal{E}(z)dz&=(2\pi)^{2|\Lambda|}\mdet^{-1/2} M\\
&\exp\Big\{\frac{1}{2}(M^{-1}(W^{1/2}h^0+W^{-1/2}(h+\zeta/\Lambda)), W^{1/2}h^0
+W^{-1/2}(h+\zeta/|\Lambda|))\Big\} .
\notag\end{align}
It is easy to see from (\ref{M}) -- (\ref{tilde_M}) that
\begin{align*}
\det \,M&= \det\,M_0(1+O(W^{-1}))=(c_+^2c_-^2a_+^2a_-^2)^{|\Lambda|} (1+O(W^{-1}))=c_0^{4|\Lambda|}(1+O(W^{-1}))
\end{align*}
with $c_\pm$ of (\ref{c_pm}). Note now that
\[
M^{-1}=\big(M_0+\dfrac{1}{W}\tilde{M}\big)^{-1}=M_0^{-1}-\dfrac{1}{W}M_0^{-1}\tilde{M}M_0^{-1}+O(W^{-2}).
\]
Since $M_0$ is diagonal and $h^0_{j,ls}$ is proportional to $n_{j,1}$ or $n_{j,2}$ and $n_{j,l}^2=0$, we have
\[(M_0^{-1}h^0,h^0)=0.\]
Hence, the exponent in the r.h.s. of (\ref{G_main1}) takes the form
\begin{multline*}
\dfrac{1}{2}\Big((M_0^{-1}h^0,h+\zeta/\Lambda)+(M_0^{-1}(h+\zeta/\Lambda),h^0)\\-(M_0^{-1}\tilde M M_0^{-1}h^0,h^0)\Big)+o(1)
=
I_1+I_2-I_3+o(1).
\end{multline*}
Then we can rewrite (recall (\ref{h}) and (\ref{capm}))
\begin{align}\label{M_0h,h}
I_1&+I_2=\sum \Big(\dfrac{(h_{j,11}+\zeta_{j,11}/|\Lambda| )n_{j,1}}{a_+c_+}+\dfrac{(h_{j,22}+\zeta_{j,22}/\Lambda )n_{j,2}}{a_-c_-}\\
&-\dfrac{(h_{j,12}+a_+\zeta_{j,12}/|\Lambda| )n_{j,1}}{a_+^2c_+}
-\dfrac{(h_{j,21}+a_-\zeta_{j,21}/|\Lambda| )n_{j,2}}{a_-^2c_-}\Big)\notag\\ 
\notag &=\sum n_{j,1}\Big(\dfrac{2}{a_+c_0}+\beta\big(t_j^2+t_{j+1}^2-v_j^2-v_{j+1}^2\big)+\dfrac{a_-n_{j,12}}{a_+^2c_0}+\dfrac{n_{j,21}}{a_-c_0}
+\dfrac{\zeta_{j,11}-\zeta_{j,12}}{c_0|\Lambda|}\Big)\\ \notag
&+\sum n_{j,2}\Big(-\dfrac{2}{a_-c_0}+\beta\big(t_j^2+t_{j+1}^2-v_j^2-v_{j+1}^2\big)-\dfrac{a_+n_{j,21}}{a_-^2c_0}-\dfrac{n_{j,12}}{a_+c_0}
-\dfrac{\zeta_{j,22}-\zeta_{j,21}}{c_0|\Lambda|}\Big)+O(W^{-1}),
\end{align}
\begin{align}\notag
&I_3=\dfrac{4}{c_0^4}\sum n_{j,1}n_{j,2}-\dfrac{1}{a_+^2c_0^2}\sum n_{j,12}n_{j,1}n_{j,2}
-\dfrac{1}{a_-^2c_0^2}\sum n_{j,21}n_{j,1}n_{j,2}\\
 \label{tilMh,h} &+\sum \dfrac{\beta(v_j^2+t_j^2)}{c_0^2}\big(n_{j,1}n_{j+1,1}+n_{j,1}n_{j+1,2}
+n_{j,2}n_{j+1,1}+n_{j,2}n_{j+1,2}\big)+O(W^{-1}).
\end{align}
Moreover,
\begin{align}\label{gr_lapl}
&\exp\Big\{\beta\sum\Tr\Big(U_j^*\tilde\rho_jS_j-U_{j-1}^*\tilde\rho_{j-1}S_{j-1}\Big) \Big(S_j^{-1}\tilde\tau_jU_j-
S_{j-1}^{-1}\tilde\tau_{j-1}U_{j-1}\Big)\Big\}\\ \notag
&=\exp\Big\{\dfrac{\beta}{W}\sum\Tr\Big(U_j^*\hat\rho_jS_j-U_{j-1}^*\hat\rho_{j-1}S_{j-1}\Big) \Big(S_j^{-1}\hat\tau_jU_j-
S_{j-1}^{-1}\hat\tau_{j-1}U_{j-1}\Big)\Big\}+O(W^{-1/2}),
\end{align}
where
\begin{align*}
\hat\rho_j=\hbox{diag}\{\rho_{j,11},\rho_{j,22}\},\quad
 \hat\tau_j=\hbox{diag}\{\tau_{j,11},\tau_{j,22}\}.
\end{align*}
Combining (\ref{M_0h,h}) -- (\ref{gr_lapl}) we can integrate the main term of (\ref{G_main1}) with respect to $\rho_{j,12}$, $\tau_{j,12}$, $\rho_{j,21}$, $\tau_{j,21}$
according to (\ref{G_Gr}).
This integration gives
\begin{align*}
&\prod\limits_{j\in\Lambda} \Big(c_++\dfrac{a_-n_{j,1}}{a_+^2c_0}-\dfrac{n_{j,2}}{a_+c_0}+\dfrac{n_{j,1}n_{j,2}}{a_+^2c_0^2}\Big)
\Big(c_-+\dfrac{n_{j,1}}{a_-c_0}-\dfrac{a_+n_{j,2}}{a_-^2c_0}+\dfrac{n_{j,1}n_{j,2}}{a_-^2c_0^2}\Big)\\ \notag
&=c_0^2+\dfrac{c_0n_{j,2}}{a_-}-\dfrac{c_0n_{j,1}}{a_+} +\big(1+2/c_0^2\big)n_{j,1}n_{j,2}=
c_0^2\cdot \exp\Big\{-\dfrac{n_{j,1}}{a_+c_0}+\dfrac{n_{j,2}}{a_-c_0}\Big\}\cdot \Big(1+\dfrac{2}{c_0^4}n_{j,1}n_{j,2}\Big),
\end{align*}
which together with (\ref{M_0h,h}) -- (\ref{gr_lapl}) yields
\begin{align*}
\mathcal{R}_{Wn\beta}^{+-}(E,\varepsilon,\xi)=& c_0^{4|\Lambda|}C_{E,\varepsilon}\int  d\hat\rho\, d\hat \tau \,dU\, d S \prod_{j\in \Lambda}
\Big(1-\dfrac{2}{c_0^4}n_{j,1}n_{j,2}\Big)\exp\Big\{-\beta c_0^2\sum (v_j^2+t_j^2)\Big\} \\ \notag
&\times\exp\Big\{\beta\sum\Tr\Big(U_j^*\hat\rho_jS_j-U_{j-1}^*\hat\rho_{j-1}S_{j-1}\Big) \Big(S_j^{-1}\hat\tau_jU_j-
S_{j-1}^{-1}\hat\tau_{j-1}U_{j-1}\Big)\Big\}\\ \notag
&\times \exp\Big\{\sum n_{j,1}\Big(\beta\big(t_j^2+t_{j+1}^2-v_j^2-v_{j+1}^2\big)+\dfrac{\zeta_{j,11}-\zeta_{j,12}}{c_0|\Lambda|}|\big)\Big)\Big\}\\
\notag&\times \exp\Big\{\sum n_{j,2}\Big(\beta\big(t_j^2+t_{j+1}^2-v_j^2-v_{j+1}^2\big)-\dfrac{\zeta_{j,22}-\zeta_{j,21}}{c_0|\Lambda|}\big)\Big)\Big\}+o(1),
\end{align*}
where we have used
\[
(1+2n_{j,1}n_{j,2}/c_0^4)\cdot e^{-4n_{j,1}n_{j,2}/c_0^4}=1-2n_{j,1}n_{j,2}/c_0^4.
\]
Now changing
\[
\rho_{j,11}\to c_0 \rho_{j,1}, \quad \tau_{j,11}\to c_0 \tau_{j,1},\quad \rho_{j,22}\to c_0 \rho_{j,2},\quad \tau_{j,22}\to c_0 \rho_{j,2}
\]
with an appropriate change in $n_{j,1}$, $n_{j,2}$, $\hat\rho_j$, $\hat\tau_j$, and recalling (\ref{beta_til}), we get
\begin{align*}
\mathcal{R}_{Wn\beta}^{+-}(E,\varepsilon,&\xi)=C_{E,\varepsilon}  \int  d\hat\rho\, d\hat \tau \,dU\, d S \prod_{j\in \Lambda}
\Big(1-2n_{j,1}n_{j,2}\Big)  \exp\Big\{-\tilde\beta\sum (v_j^2+t_j^2)\Big\}\\ 
\notag
&\times\exp\Big\{\tilde\beta\sum\Tr\Big(U_j^*\hat\rho_jS_j-U_{j-1}^*\hat\rho_{j-1}S_{j-1}\Big) \Big(S_j^{-1}\hat\tau_jU_j-
S_{j-1}^{-1}\hat\tau_{j-1}U_{j-1}\Big)\Big\}\\ \notag
&\times \exp\Big\{\sum n_{j,1}\Big(\tilde\beta\big(t_j^2+t_{j+1}^2-v_j^2-v_{j+1}^2\big)+c_0(\zeta_{j,11}-\zeta_{j,12})/|\Lambda|\big)\Big)\Big\}\\
\notag&\times \exp\Big\{\sum n_{j,2}\Big(\tilde\beta\big(t_j^2+t_{j+1}^2-v_j^2-v_{j+1}^2\big)-c_0(\zeta_{j,22}-\zeta_{j,21})/|\Lambda|\big)\Big)\Big\}\\ \notag
&\times\exp\Big\{\dfrac{ic_0}{2|\Lambda|}\sum\limits_{j\in\Lambda}\Big(\Tr U_j^*L U_j\big(i\varepsilon L+\hat{\xi}/\rho(E)\big)+\Tr S_j^{-1} LS_j\big(i\varepsilon L+
\hat{\xi}'/\rho(E)\big)\Big)\Big\},
\end{align*}
which can be rewritten as (\ref{sigma-mod}). The second
relation of  (\ref{sigma-mod})  follows from the uniform in $\xi$ convergence  of $\mathcal{R}_{Wn\beta}^{+-}(E,\varepsilon,\xi)$, as $W\to\infty$.

$\square$

\subsection{Proof of  Theorem \ref{t:2}.}
 Theorem \ref{t:2} can be proved in a similar way. First of all we can write an analogue of (\ref{G_last}):
\begin{align}\label{G+_last}
&\mathcal{R}_{Wn\beta}^{++}(E,\varepsilon,\xi)=Q^{(2)}_{W, |\Lambda|}\displaystyle
\int  dVdU
\int  d x\int_{\mathbb{R}_+^{2|\Lambda|}}  d b\prod_{j\in\Lambda}\dfrac{(x_{j,1}-x_{j,2})^2
(b_{j,1}-b_{j,2})^2}{b_{j,1}^2b_{j,2}^2}\\ \notag
&\times 
\exp\Big\{-W\sum\limits_{j\in\Lambda}\sum\limits_{\sigma=1}^2\left(f(x_{j,\sigma})+f(b_{j,\sigma})\right)\Big\}
\cdot \mdet \,\mathcal{D}(\hat{X}, \hat{B}, U, V) \\ \notag
&\times \exp\Big\{
\dfrac{\beta}{2}\sum\limits_{j\sim k}\Big(\Tr (V_j^*\hat{B}_jV_j-V_k^*\hat{B}_kV_k)^2-\Tr (U_j^*\hat{X}_jU_j-U_k^*\hat{X}_kU_k)^2\Big)\Big\}\\ \notag
&\times\exp\Big\{\dfrac{1}{2W}\sum\limits_{j, k}
R_{jk}\,\Tr (U_j^*\hat{X}_jU_j-iZ_1)(U_k^*\hat{X}_kU_k-iZ_1)\Big\} \\ \notag
&\times\exp\Big\{\dfrac{i}{|\Lambda|}\sum\limits_{j\in\Lambda}\Big(\Tr U_j^*\hat{X}_jU_j\big(i\varepsilon\cdot I+\hat{\xi}_1/\rho(E)\big)+\Tr V_j^*\hat{B}_jV_j\big(i\varepsilon\cdot I+\hat{\xi}_2/\rho(E)\big)\Big)\Big\}.
\end{align}
Note that (\ref{G+_last}) has unitary $V_j$ instead of hyperbolic $S_j$ and $i\varepsilon\cdot I$ instead of $i\varepsilon\cdot L$.
Then we deform the contours of integration as
\begin{itemize}
\item for $x_{j,1}$, $x_{j,2}$ to $iE/2+\mathbb{R}$;
\item for $b_{j,1}, b_{j,2}$ to $\mathcal{L}_+(E)$ of (\ref{L_cal})
\end{itemize}
and prove the following lemma in the same way as Lemma \ref{l:s_point}):
\begin{lemma}\label{l:s_point+}
The integral (\ref{G+_last}) over $\{x_{j,l}\}, \{b_{j,l}\}$, $l=1,2$, $j\in \Lambda$ can be restricted to the integral over the neighbourhood of the points
\begin{itemize}
\item[I.] $x_{j,1}=a_+$, $x_{j,2}=a_-$ or $x_{j,1}=a_-$, $x_{j,2}=a_+$, $b_{j,1}=b_{j,2}=a_+$ for any $j\in\Lambda$;
\item[II.] $x_{j,1}=x_{j,2}=a_+$, $b_{j,1}=b_{j,2}=a_+$ for any $j\in\Lambda$;
\item[III.] $x_{j,1}=x_{j,2}=a_-$, $b_{j,1}=b_{j,2}=a_+$ for any $j\in\Lambda$.
\end{itemize}
Moreover, the contributions of the points I and II are $o(1)$, as $W\to\infty$.
\end{lemma}
Indeed, the contribution of the point II is small, since after an appropriate change of variables similar to (\ref{change1}) (which gives $W^{-2|\Lambda|}$)
the expression 
\[
(x_{j,1}-x_{j,2})^2(b_{j,1}-b_{j,2})^2
\]
gives $W^{-2|\Lambda|}$, and the expansion of $\mdet\,\mathcal{D}(\hat{X}, \hat{B}, U, V)$ starts from  $W^{-2|\Lambda|}$ 
(see (\ref{d_exp1})). 

For the points I the expression for $\mdet\,\mathcal{D}(\hat{X}, \hat{B}, U, V)$ starts from  $W^{-|\Lambda|}$, and
another $W^{-|\Lambda|}$ comes from $(b_{j,1}-b_{j,2})^2$. Therefore similarly to (\ref{++}) we get that the main contribution
around these saddle-points is given by
\begin{align}\label{b++}
C\cdot \Big\langle\prod\limits_{j\in\Lambda}\big(\tilde{x}_{j,1}/a_+-\tilde{b}_{j,1}\big)\cdot \big(\tilde{x}_{j,1}/a_+-\tilde{b}_{j,2}\big) \cdot (\tilde{b}_{j,1}
-\tilde{b}_{j,2})^2\Big\rangle+o(1),
\end{align}
where
\begin{align*}
\Big\langle \cdot \Big\rangle=\int \Big(\cdot\Big) \exp\Big\{-\frac{1}{2}\sum_{j\in\Lambda}\big({c_+\tilde x_{j,1}^2+c_-\tilde x_{j,2}^2
+a_+^2c_+ (\tilde b_{j,1}^2+
\tilde b_{j,2}^2})\big)\Big\}  d\tilde{x}d\tilde{b}.
\end{align*}
But it is easy to see that  the Gaussian integral in (\ref{b++}) is zero.

Thus we are left to compute the contribution of the point III. Doing again an appropriate change of variables similar to (\ref{change1}) , we see that the expression
\[
(x_{j,1}-x_{j,2})^2(b_{j,1}-b_{j,2})^2
\]
already gives $W^{-2|\Lambda|}$, and hence to obtain a non-zero contribution we have to compute
\begin{align*}
&\int \prod\limits_{j\in\Lambda} (\tilde x_{j,1}-\tilde x_{j,2})^2(\tilde b_{j,1}-\tilde b_{j,2})^2 \exp\Big\{-\frac{1}{2}\sum_{j\in\Lambda}\big({c_-(\tilde x_{j,1}^2+-\tilde x_{j,2}^2)
+a_+^2c_+ (\tilde b_{j,1}^2+
\tilde b_{j,2}^2})\big)\Big\}  d\tilde{x}d\tilde b\\
&=\Big((2\pi)^{2}\cdot 4 (c_+c_-a_+)^{-2}\Big)^{|\Lambda|}
\end{align*}
and take only zero terms in  the expansions of all other functions in (\ref{G+_last}). That gives the first relation of (\ref{G++_lim}). The second
relation of (\ref{G++_lim}) follows from the uniform in $\xi$ convergence  of $\mathcal{R}_{Wn\beta}^{++}(E,\varepsilon,\xi)$ as $W\to\infty$.
$\Box$

\section{Proof of Theorem \ref{thm:cor}}\label{s:4}
According (\ref{cor=det}), (\ref{G_2}),  (\ref{sigma-mod}), and (\ref{G++_lim}), to prove Theorem \ref{thm:cor}, it is sufficient to show that
\begin{align}\label{main}
(2\pi)^{-2}\lim_{\varepsilon\to 0}\lim_{\beta,n\to \infty}  \dfrac{\partial^2}{\partial \xi_1' \partial \xi_2'}&\Big(
\mathcal{R}_{n\beta}^{+-}(E,\varepsilon,\xi)+
\overline{\mathcal{R}_{n\beta}^{+-}}(E,\varepsilon,\xi)\\ \notag
&-\mathcal{R}_{n\beta}^{++}(E,\varepsilon,\xi)-\overline{\mathcal{R}_{n\beta}^{++}}(E,\varepsilon,\xi)\Big)\Big|_{\xi^\prime=\xi}=1-\dfrac{\sin^2 (\pi(\xi_1-\xi_2))}
{\pi^2(\xi_1-\xi_2)^2}.
\end{align}
Using (\ref{G++_lim})  we get
\begin{align}\label{lim_++}
\lim_{\varepsilon\to 0}\,\lim_{\beta,n\to \infty}\dfrac{\partial^2}
{\partial\xi_1^\prime\partial\xi_2^\prime}\left(\mathcal{R}_{n\beta}^{++}(E,\varepsilon,\xi)+\overline{\mathcal{R}_{n\beta}^{++}}(E,\varepsilon,\xi)\right)\Big|_{\xi^\prime=\xi}
=-\dfrac{a_+^2+a_-^2}{\rho^2(E)}.
\end{align}
In addition,   $\mathcal{R}_{n\beta}^{+-}(E,\varepsilon,\xi)$  are analytic functions in any of $\xi_1,\xi_2,\xi_1',\xi_2'$ for $\Im\xi_1',\Im\xi_1'>-\varepsilon$,
and they are uniformly bounded in $n,\beta$  for $\xi_1,\xi_2,\xi_1',\xi_2'$ varying  in any compacts   satisfying this condition. Hence, we can 
replace the order of the derivative and the limiting transition and by (\ref{t1.1}) obtain
\begin{align*}
&\lim_{\beta,n\to \infty}\dfrac{\partial^2}
{\partial\xi_1^\prime\partial\xi_2^\prime} \mathcal{R}_{n\beta}^{+-}(E,\varepsilon,\xi)\Big|_{\xi^\prime=\xi}\\=
&\dfrac{\partial^2}
{\partial\xi_1^\prime\partial\xi_2^\prime} C_{E,\varepsilon}
  e^{-c_0(\alpha_1+\alpha_2)}\Big(\delta_1\delta_2(e^{2c_0\alpha_1}-1)/\alpha_1\alpha_2
-(\delta_1+\delta_2)e^{2c_0\alpha_1}/\alpha_2+e^{2c_0\alpha_1}\alpha_1/\alpha_2\Big)\Big|_{\xi^\prime=\xi}.
\end{align*}
Computing the derivative,  we get
\begin{align*}
\lim_{\beta,n\to \infty} \dfrac{\partial^2}
{\partial\xi_1^\prime\partial\xi_2^\prime}\mathcal{R}_{n\beta}^{+-}(E,\varepsilon,\xi)\Big|_{\xi^\prime=\xi}
=\dfrac{1}{\rho^2(E)}-\dfrac{1-e^{2\pi i\theta_\varepsilon}}{\theta_\varepsilon^2},
\end{align*}
where
\[
\theta_\varepsilon=2i\alpha_1\rho(E)=2i\varepsilon\rho(E)+\xi_1-\xi_2.
\]
This yields 
\begin{align*}
\lim_{\beta,n\to \infty}\dfrac{\partial^2}
{\partial\xi_1^\prime\partial\xi_2^\prime}\left(\mathcal{R}_{n\beta}^{+-}(E,\varepsilon,\xi)+\overline{\mathcal{R}_{n\beta}^{+-}}(E,\varepsilon,\xi)\right)\Big|_{\xi^\prime=\xi}
=\dfrac{2}{\rho^2(E)}+\dfrac{(e^{i\pi\theta_\varepsilon}-e^{-i\pi \theta_\varepsilon})^2}{\theta_\varepsilon^2},
\end{align*}
and hence
\begin{align*}
\lim_{\varepsilon\to 0}\,\lim_{\beta,n\to \infty} \dfrac{\partial^2}
{\partial\xi_1^\prime\partial\xi_2^\prime}\left(\mathcal{R}_{n\beta}^{+-}(E,\varepsilon,\xi)+\overline{\mathcal{R}_{n\beta}^{+-}}(E,\varepsilon,\xi)\right)\Big|_{\xi^\prime=\xi}
=\dfrac{2}{\rho^2(E)}-\dfrac{4\sin^2(\pi(\xi_1-\xi_2))}{(\xi_1-\xi_2)^2},
\end{align*}
which combined with (\ref{lim_++}), and
\[
a_+^2+a_-^2+2=(a_+-a_-)^2=4\pi^2\rho(E)^2,
\]
gives (\ref{main}), thus Theorem \ref{thm:cor}.

\section{Proof of Theorem \ref{t:1}}\label{s:5}

Let us note that  to prove Theorem \ref{t:1}, it suffices to prove it only for $\xi$ such that 
\begin{align}\label{cond_xi}
&\Re\xi_1=\Re\xi_2,\quad\Re\xi_1'=\Re\xi_2',\quad \xi_1,\xi_2,\xi_1',\xi_2'\in\Omega_{c\varepsilon}\\
&\Omega_{c\varepsilon}=\{\xi:\Im\xi>-c\varepsilon\},\quad (0<c<1).
\notag\end{align}
 Indeed, assume that  $\{\mathcal{R}_{n\beta}^{+-}(E,\varepsilon,\xi)\}$ are uniformly bounded in  $n,\beta$  
for $\xi_1,\xi_2,\xi_1',\xi_2'\in\Omega_{c\varepsilon}$. Consider
$\{\mathcal{R}_{n\beta}^{+-}(E,\varepsilon,\xi)\}$  as  functions on $\xi_1$ with fixed $\xi_2,\xi_1',\xi_2'$ such that $\Re\xi_1'=\Re\xi_2'$.
Since these functions are analytic in $\Omega_{c\varepsilon}$,  the standard complex analysis argument yields that (\ref{t1.1})  on the segment
$\Re\xi_1=\Re\xi_2$ implies (\ref{t1.1})  for any $\xi_1\in\Omega_{c\varepsilon}$, hence for any $\xi_1,\xi_2\in \Omega_{c\varepsilon}$.
Then, fixing any $\xi_1,\xi_2,\xi_2'$, we can consider $\{\mathcal{R}_{n\beta}^{+-}(E,\varepsilon,\xi)\}$
as a sequence of analytic functions on $\xi_1'$. Since, by the above argument, (\ref{t1.1}) is valid on the segment $\Re\xi_1'=\Re\xi_2'$, 
the same argument yields that (\ref{t1.1}) is valid for any $\xi_1',\xi_2'$. Therefore, it is enough to prove Theorem \ref{t:1} for real $\alpha_1>\varepsilon/2$, 
$\alpha_2>\varepsilon/2$, which means that we take $c=\rho(E)$  (see the definition (\ref{alp})).

To check that $\{\mathcal{R}_{n\beta}^{+-}(E,\varepsilon,\xi)\}$ are uniformly bounded in  $n,\beta$  for $\xi_1,\xi_2,\xi_1',\xi_2'\in\Omega_{c\varepsilon}$, we apply the  Schwartz inequality to $\mathcal{R}_{Wn\beta}^{+-}(E,\varepsilon,\xi)$ in the form (\ref{G_2}). Then we get
\begin{align*}
&|\mathcal{R}_{Wn\beta}^{+-}(E,\varepsilon,\xi)|^2\le |\mathcal{R}_{Wn\beta}^{+-}(E,\varepsilon,\xi_1)|\,|\mathcal{R}_{Wn\beta}^{+-}(E,\varepsilon,\xi_2)|\\
\Rightarrow &\mathcal{R}_{n\beta}^{+-}(E,\varepsilon,\xi)|^2\le |\mathcal{R}_{n\beta}^{+-}(E,\varepsilon,\xi_1)|\,|\mathcal{R}_{n\beta}^{+-}(E,\varepsilon,\xi_2)|
\end{align*}
where $\xi_1=(\xi_1,\xi_1,\xi_1',\xi_1')$, $\xi_2=(\xi_2,\xi_2,\xi_2',\xi_2')$. Since $\xi_1,\xi_2$ satisfy (\ref{cond_xi}), the uniform boundedness of
the r.h.s. follows from the uniform convergence (in $\xi$, satisfying (\ref{cond_xi})) of (\ref{t1.1}) (see Section \ref{ss:t1}).

\subsection{Representation of $\mathcal{R}_{n\beta}^{+-}$ in the operator form}
Now we are going to represent $\mathcal{R}_{n\beta}^{+-}$ in 1d case in the operator form. Put $n=|\Lambda|$, and set
\begin{equation*}
{\mathcal{M}}(Q,Q')=\mathcal{F}(Q)H(Q,Q') \mathcal{F}(Q'),
\end{equation*}
where 
\begin{align}\label{H_op}
&H(Q,Q')=\exp\Big\{\dfrac{\tilde\beta}{4}\Str QQ'\Big\} (1-n_1n_2) (1-n_1'n_2')
\\ \notag
&\mathcal{F}(Q)=\exp\Big\{-\dfrac{c_0}{4n}\Str Q \Lambda_{\xi, \varepsilon}\Big\}=F(U,S)\cdot \exp\Big\{n_1\cdot F_1(U,S)+n_2\cdot F_2(U,S)\Big\} 
\end{align}
with $Q$, $Q'$ of the form (\ref{Q}) and
\begin{align}\label{F}
&F(U,S)=\exp\big\{-\dfrac{c_0}{n}\big(\alpha_1(1-|U_{12}|^2)
+\alpha_2\cdot |S_{12}|^2\big)\big\},\\ \notag
&F_1(U,S)=-{c_0}\big(\delta_1-\alpha_1 \cdot |U_{12}|^2-\alpha_2\cdot |S_{12}|^2\big)/n,\\ \notag
&F_2(U,S)=-{c_0}\big(\delta_2-\alpha_1 \cdot |U_{12}|^2-\alpha_2\cdot |S_{12}|^2\big)/n,\\
\notag 
&n_l=\rho_l\tau_l,\quad n_l'=\rho'_l\tau'_l,\quad l=1,2,
\end{align}
and $\alpha_{1,2},\delta_{1,2}$  defined in (\ref{alp}).  Hence, by (\ref{sigma-mod})
\begin{align}\label{trans}
\mathcal{R}_{n\beta}^{+-}(E,\varepsilon,\xi)=&C_{E,\varepsilon}e^{c_0(\alpha_1-\alpha_2)}
\int  (1-n_1n_2) \mathcal{F}(Q){\mathcal{M}}^{n-1}(Q,Q')\mathcal{F}(Q') (1-n_1'n_2') dQ dQ'
\end{align}
with
\[
dQ=d\rho_1d\tau_1d\rho_2d\tau_2 dU dS.
\]
Note that $\mathcal{M}$, $H$, $\mathcal{F}$  can be considered as operators acting on the space of polynomials 
of Grassmann variables $\rho_l'$, $\tau_l'$,
$l=1,2$ with coefficients from $L_2(U)\otimes L_2(S)$, where $L_2$ are taken with respect to the Haar measures on $\mathring U(2)$, $\mathring U(1,1)$.
It is easy to see these that operators transform any even Grassmann polynomial into an even polynomial and an odd one into an odd one. In addition, they preserve
the modulo of the difference between the number of $\rho_{l}$ and the number of $\tau_{l}$. Since we are going to apply these operators only to  even polynomials which contain
 equal numbers of $\rho_{l}$ and  $\tau_{l}$, we need to study a restriction of $\mathcal{M}$, $H$, $\mathcal{F}$ to the space  
 $\mathcal P_6\cong (L_2(U(2))\otimes L_2(U(1,1)))^6$ of polynomials
\begin{equation}\label{P_6}
\widehat q=q_0+q_1n_1'+q_2n_2'+q_3n_1'n_2'+q_4\rho_1'\tau_2'+q_5\rho_2'\tau_1'.
\end{equation}
Thus  $\mathcal{M}$ is represented by a $6\times 6$ matrix $\mathcal P_6\to \mathcal P_6$ (which we also denote $\mathcal{M}$) of the form 
$\mathcal{ F} H\mathcal{ F}\big|_{\mathcal P_6}$, 
 the entries of the matrix $H$ are the integral operators on $L_2(U)\otimes L_2(S)$ with the kernels of the form $v(U(U')^*,S(S')^{-1})$ (the  integrals  are taken  with respect to $dU'dS'$), and the entries of the matrix $\mathcal{F}$ are operators of multiplication in $L_2(U)\otimes L_2(S)$. Then (\ref{trans}) takes the form
\begin{align}\label{trans1}
&\mathcal{R}_{n\beta}^{+-}(E,\varepsilon,\xi)= C_{E,\varepsilon}e^{c_0(\alpha_1-\alpha_2)}
\int (\mathcal{M}^{n-1}\tilde f(U',S'),\tilde g(U,S))_{6}dUdS dU'dS',\\
\notag
&\tilde f(U,S):= \mathcal{F}\cdot (1-n_1n_2),\quad
\tilde g(U,S):=\mathcal{F}\cdot (1-n_1n_2),
\end{align}
where by $(\cdot,\cdot)_{6}$ we mean the "scalar" product in $\mathcal P_6$ which 
gives the coefficient in front of $n_1n_2$ in the product of two polynomials of the form (\ref{P_6}).

\subsection{Proof of Theorem \ref{t:1} for $\alpha_1,\alpha_2>\varepsilon/2$ }\label{ss:t1}
As it was mentioned in the beginning of Section 5, it suffices to prove Theorem \ref{t:1} for real $\alpha_1,\alpha_2>\varepsilon/2$.

The proof of (\ref{t1.1}) is based on the following representation of $\mathcal{R}_{n\beta}^{+-}(E,\varepsilon,\xi)$.
\begin{proposition}\label{p:repr} For any $\xi$ such that  $\alpha_1,\alpha_2>\varepsilon/2$ (see (\ref{alp})) we have
\begin{align}\label{repr1}
&\mathcal{R}_{n\beta}^{+-}(E,\varepsilon,\xi)=\frac{C_{E,\varepsilon}e^{c_0(\alpha_1-\alpha_2)}}{2\pi i}\oint_{\omega_A}z^{n-1}(\widehat G(z)\widehat f,\widehat g)dz,\quad \omega_A=\{z:|z|=1+A/n\},\\ \label{M_hat}
&  \widehat G(z)=(\widehat M-z)^{-1},\quad\widehat M=\widehat F\widehat K\widehat F,\quad \widehat K=\widehat K_0+O(\beta^{-1}),
\end{align}
where operators $\widehat K_0$, $\widehat F$  and the vectors $\widehat f$, $\widehat g$ have  the form
\begin{align}\label{repr2}
&\quad \widehat K_0=\left(\begin{array}{cccc}K_{US}&\widetilde K_1&\widetilde K_2&\widetilde K_3\\
0&K_{US}&0&\widetilde K_2\\0&0&K_{US}&\widetilde K_1\\0&0&0&K_{US}\end{array}\right),\quad 
\widehat F=F\left(\begin{array}{cccc}1&F_1 &F_2&F_1 F_2\\
0&1&0&F_2\\0&0&1&F_1\\0&0&0&1\end{array}\right)\\ \notag
&\hat f=\widehat F (e_4-  e_1),\quad\hat g=\widehat F^{(t)}(e_1- e_4)
\end{align}
with  $F$  and $ F_{1,2}$  being the  operator on $L_2(U)\otimes L_2(S)$ of multiplication by the functions $F$  and $F_{1,2}$ 
defined in (\ref{F}), $K_{US}=K_U\otimes K_S$ and
$K_U$ and $K_S$  being the integral operators in $L_2(U)$ and $L_2(S)$  with a "difference" kernels
\begin{align*}
K_U(U,U')=K_U(U(U')^*)=\tilde{\beta}  e^{-\tilde{\beta} |(U(U')^*)_{12}|^2},\\
K_S(S,S')=K_S(S(S')^{-1})=\tilde{\beta}  e^{-\tilde{\beta} |(S(S')^{-1})_{12}|^2}.
\notag\end{align*}
Here $\widetilde K_p$, $p=1, 2, 3$  are normal operators on $L_2(U)\otimes L_2(S)$, they commute with $K_{US}$ and with the Laplace operators $\widetilde\Delta_U,\widetilde\Delta_S$ 
on the corresponding groups
and satisfy the bounds
\begin{align}\label{bound_re}
&|\widetilde K_p|\le C(1-K_{US})\le -C(\widetilde\Delta_U+\widetilde\Delta_S)/\beta,
\end{align}
where   the Laplace operators $\widetilde\Delta_U,\widetilde\Delta_S$ for the functions depending only on $|S_{12}|^2$ and
$|U_{12}|^2$ have the form
\[\widetilde\Delta_S(\varphi)=-\frac{d}{dx} x(x+1)\frac{d\varphi}{dx} \quad (x=|S_{12}|^2),\qquad 
\widetilde\Delta_U(\varphi)=-\frac{d}{dx} x(1-x)\frac{d\varphi}{dx} \quad (x=|U_{12}|^2).\]
\end{proposition}
We postpone the proof of the proposition to Section \ref{s:6} and now derive (\ref{t1.1}) from it.
To this end, set  
\begin{align*}
\widehat M_0=\widehat F^2,\quad \widehat G_0=(\widehat M_0-z)^{-1},
\end{align*}
and consider
\begin{align*}
\Delta G:=\widehat G-\widehat G_0=-\widehat G_0(\widehat M-\widehat M_0)\widehat G_0-\widehat G_0(\widehat M-\widehat M_0)\widehat G(\widehat M-\widehat M_0)\widehat G_0.
\end{align*}
We  apply the following lemma, which we will prove later:
\begin{lemma}\label{l:bG}
For any $z\in\omega_A$ (see (\ref{repr1})) we have the bounds
\begin{align}\label{b_1}
&\|(\widehat M-\widehat M_0)\widehat G_0\widehat f\|^2\le C(n/\tilde{\beta})^2,\quad \|(\widehat M-\widehat M_0)\widehat G_0\widehat g\|^2\le C(n/\tilde{\beta})^2\\
&|(\widehat G_0(\widehat M-\widehat M_0)\widehat G_0\widehat f,\widehat g)|\le n\tilde{\beta}^{-1}/|z-1|,\quad
 \|\widehat G\|\le C\log^2n/|z-1|.\quad 
\notag\end{align} 
\end{lemma}
\noindent The lemma implies that
\begin{align*}
&\Big|\frac{1}{2\pi i}\oint_{\omega_A}z^{n-1}(\Delta G\widehat f,\widehat g)dz\Big|\le C\oint_{\omega_A} |(\widehat G_0(\widehat M-\widehat M_0)\widehat G_0\widehat f,\widehat g)|\,|dz|\\
&+ C\oint_{\omega_A}\|\widehat G(z) \|\cdot \|(\widehat M-\widehat M_0)\widehat G_0(z)\widehat f\|\cdot
\|(\widehat M-\widehat M_0)\widehat G_0(\bar z)\widehat g\|\,|dz|\\
&\le C(n/\tilde{\beta})\oint_{\omega_A}\frac{|dz|}{|z-1|}\le Cn\log n/\tilde{\beta}\to 0,
\end{align*}
where we used $n\log^2 n\ll \tilde\beta$ and
\[
\oint_{\omega_A}\frac{|dz|}{|z-1|}\le C\log n.
\]
Thus we have proved that (recall (\ref{repr2}))
\begin{align*}
\mathcal{R}_{n\tilde{\beta}}^{+-}(E,\varepsilon,\xi)&=\frac{C_{E,\varepsilon}e^{c_0(\alpha_1-\alpha_2)}}{2\pi i}\oint_{\omega_A}z^{n-1}
(\widehat G_0(z)\widehat f,\widehat g)dz+o(1)=C_{E,\varepsilon}e^{c_0(\alpha_1-\alpha_2)}(\widehat F^{2n-2}\widehat f,\widehat g)+o(1)\\
&=C_{E,\varepsilon}e^{c_0(\alpha_1-\alpha_2)}\int \big(4n^2F_1F_2-2) F^{2n}dUdS+o(1).
\end{align*}
Performing the integration with respect to $dU$, $dS$ we obtain (\ref{t1.1}). $\square$

\smallskip

\textit{Proof of Lemma \ref{l:bG}}. To prove the first inequality of (\ref{b_1}), observe that since $\widehat F$ is bounded we have
\[\|(\widehat M-\widehat M_0)\widehat G_0\widehat f\|^2=\|\widehat F(\widehat K-1)\widehat F\widehat G_0\widehat f\|^2\le C \|(\widehat K_0-1)\widehat F\widehat G_0\widehat f\|^2.
\]
Moreover, since $\widetilde K_\alpha$ and $1-K_{US}$ commute with $\widetilde\Delta_U,\widetilde\Delta_S$,  (\ref{bound_re}) implies
\begin{align*}
&(\widehat K_0-1)^*(\widehat K_0-1)\le C\tilde{\beta}^{-2}(\widetilde\Delta_U+\widetilde\Delta_S)^2\\
\Rightarrow& \|(\widehat M-\widehat M_0)\widehat G_0\widehat f\|^2\le C\tilde{\beta}^{-2} \|(\widetilde\Delta_U+\widetilde\Delta_S)\widehat F\widehat G_0\widehat f\|^2\le
C'\tilde{\beta}^{-2}( \|\widetilde\Delta_U\widehat G_0\widehat F\widehat f\|^2+ \|\widetilde\Delta_S\widehat G_0\widehat F\widehat f\|^2)\\
&\le C'\tilde{\beta}^{-2}\max_{\mu,\nu\le 4}( \|\widetilde\Delta_U(\widehat G_0)_{\mu\nu}(\widehat F\widehat f)_\nu\|^2+ ( \|\widetilde\Delta_S(\widehat G_0)_{\mu\nu}(\widehat F\widehat f)_\nu\|^2).
\end{align*}
It is easy to see that $\widehat G_0$ has  the same form as the matrices in  (\ref{repr2})  with zeros below the main diagonal and
\begin{align*}
(\widehat G_0)_{ii}=G_0:=(F^2-z)^{-1},\quad (\widehat G_0)_{23}=0,\quad (\widehat G_0)_{12}=(\widehat G_0)_{34}=-2F_1G_0^2F^2\\
(\widehat G_0)_{13}=(\widehat G_0)_{24}=-2F_1G_0^2F^2,\quad (\widehat G_0)_{14}=8F_1F_2G_0^3F^4-4F_1F_2G_0^2F^2
\end{align*}
(recall that here all operators  commute with each other). In addition $(\widehat F\widehat f)_\nu,\,\nu=1,..,4$ are the linear combinations of the
functions
$ (F_1)^{\gamma_1} (F_2)^{\gamma_2}F^\sigma$ with $\gamma_{1,2}=0,1,2$, $\sigma=1,2$.
Let us estimate  the term which appears after the application of $\widetilde\Delta_SF^4F_1F_2G_0^3$ to the function
$F^2$ (the other terms can be estimated similarly). Rewrite
\begin{align*}
&\tilde{\beta}^{-2}\|\widetilde \Delta_SF_1 F_2(F^2-z)^{-3}F^6\|^2\\
=&C\tilde{\beta}^{-2}\int_0^\infty dx\Big|\frac{d}{dx}(x^2+x)\frac{d}{dx}\frac{(x+c_1)(x+c_2)}{n^2}\frac{e^{-3\alpha x/n}}{(e^{-\alpha x/n}-z)^3}\Big|^2,
\end{align*}
where $c_1$ and $c_2$ correspond to the terms of (\ref{F}), which do not depend on $x=|S_{12}|^2$, end $\alpha=2c_0\alpha_2>0$.
Changing $\tilde x=x/n$ we get
\begin{align*}
&\tilde{\beta}^{-2}n\int_0^\infty d\tilde x\Big|\frac{d}{d\tilde x}\tilde x(\tilde x+1/n)\frac{d}{d\tilde x}(\tilde x+c_1/n)
(\tilde x+c_2/n)\frac{e^{-3\alpha\tilde x}}{(e^{-\alpha\tilde x}-z)^3}\Big|^2\\
\le &C\tilde{\beta}^{-2}n\int_0^\infty d\tilde x\Big|\frac{(\tilde x+c/n)^{2}}{|e^{-\alpha\tilde x}-z|^3}+\frac{(\tilde x+c/n)^{3}}{|e^{-\alpha\tilde x}-z|^{4}}
+\frac{(\tilde x+c/n)^{4}}{|e^{-\alpha\tilde x}-z|^{5}}\Big|^2e^{-6\alpha\tilde x}\\
\le &C\tilde{\beta}^{-2}n\int_0^\infty \frac{e^{-6\alpha\tilde x}}{|e^{-\alpha\tilde x}-z|^2}d\tilde x\le C(n/\tilde{\beta})^2, \quad |z|\ge 1+A/n.
\end{align*}
Here $c=\max\{|c_1|,|c_2|,1\}$.

The second and the third inequality in (\ref{b_1}) can be obtained similarly.

To obtain the bound for $\|\widehat G\|$, we introduce 
\[\widehat M_1:=\widehat F\widehat K_0\widehat F,\quad \widehat G_1:=(\widehat M_1-z)^{-1}\]
and prove that
\begin{align*}
\|\widehat G_1\|\le C\log^2n/|z-1|,
\end{align*}
or, equivalently,
\begin{align}\label{b_G1}
\|\widehat G_{1,ij}\|\le C\log^2n/|z-1|.
\end{align}
Observe that $\widehat M_1$ have the same form as the matrices in  (\ref{repr2})  with
$K_{US}\to FK_{US}F$,  $\widetilde K_i\to L_i$, where
\begin{align}\label{L}
L_1=&FK_{US}FF_1+F_1FK_{US}F+F\widetilde K_1 F,\quad \\ \notag
L_2=&FK_{US}FF_2+F_2FK_{US}F+F\widetilde K_2 F\\ \notag
L_3=&F\widetilde K_3F+F_1F_2FK_{US}F+FK_{US}FF_1F_2+F_1FK_{US}FF_2+F_2FK_{US}FF_1\\ \notag
&+F\widetilde K_1FF_2+F_1F\widetilde K_2F+F\widetilde K_2FF_1+F_2F\widetilde K_1F.
\end{align}
Then the matrix $\widehat G_1:=(\widehat F\widehat K_0\widehat F-z)^{-1}$ has zeros at the same places as in (\ref{repr2}) and
\begin{align*}
&\widehat G_{1,ii}=G:=(FK_{US}F-z)^{-1},\quad \widehat G_{1,1i}=\widehat G_{1,(4-i)4}=-GL_{i-1}G,\quad i=2,3,\\
&\widehat G_{1,14}=GL_1GL_2G+GL_2GL_1G-G L_3G,
\end{align*}
Since the spectrum of  $FK_{US}F$ belongs to  $[0,1]$, it is evident that
\begin{equation}\label{Gii}
\|G_{1,ii}\|=\|G\|\le C/|z-1|.
\end{equation}
To estimate the non-diagonal entries, we
set
\[G_*:=G(z)\Big|_{z=1+A/n}\]
and prove  the bounds
\begin{align}\label{b_GKG}
&\|G^{1/2}_*F\widetilde K_{\alpha}FG_*^{1/2}\|\le \|G^{1/2}_*F|\widetilde K_{\alpha}|FG_*^{1/2}\|\le C,\quad
 \|G_*^{-1/2}G^{1/2}\|\le C,\quad \alpha=1,2,3,\\
 &\|G^{1/2}_*F_{\alpha}FK_{US}FG_*^{1/2}\|\le C\log n,\quad \alpha=1,2\label{b_GKG1}\\
& \|G^{1/2}_*F_{1}F_{2}FK_{US}FG_*^{1/2}\|\le C\log^2 n,
\quad\|G^{1/2}_*F_{1}FK_{US}FF_{2}G_*^{1/2}\|\le C\log ^2n\notag\\
& \|G^{1/2}_*F_{\alpha}^2FG_*^{1/2}\|\le C\log^2 n,\quad \alpha=1,2.
\notag \end{align}
It is easy to see from (\ref{L}) that $GL_{1,2}G$, $GL_{1}GL_2G$  and $GL_{2}GL_1G$ can be represented as a linear combination of the terms
 $G^{1/2}\Pi G^{1/2}$, where $\Pi$ is some product of the operators whose bounds are given in  (\ref{b_GKG}) and
 the first  line of (\ref{b_GKG1}) or operators similar to them (e.g., $G^{1/2}_*FK_{US}FF_{\alpha}G_*^{1/2}$ instead of $G^{1/2}_*F_{\alpha}FK_{US}FG_*^{1/2}$,
 etc.). For instance,
\[
GFK_{US}FF_1G=G^{1/2} \cdot (G^{1/2}G_*^{-1/2})\cdot (G_*^{1/2}FK_{US}FF_1G_*^{1/2})\cdot (G_*^{-1/2}G^{1/2}) \cdot G^{1/2}.
\] 
Therefore (\ref{Gii}) and  the first line of (\ref{b_GKG1}) yield
\[\|GL_{1,2}G\|\le C\log n \cdot \|G\|\le C\log n/|z-1|,\quad \|GL_{1}GL_2G\|+ \|GL_{1}GL_2G\|\le C\log^2 n/|z-1|.\]
To estimate $GL_3G$, we use the bounds from the last two lines of (\ref{b_GKG1}), combined with the  inequality (recall that $G_*$ and $F$ are self-adjoint,
and $F$ commutes with $F_2$)
\begin{align}\label{est}
\|G^{1/2}_*F\widetilde K_1FF_2G^{1/2}_*\|\le &\|G^{1/2}_*F\widetilde K_1F\widetilde K_1^* FG^{1/2}_*\|^{1/2}\cdot \|G^{1/2}_*F_2F \bar F_2G^{1/2}_*\|^{1/2}\\ \notag
\le &\|G^{1/2}_*F|\widetilde K_1 |FG^{1/2}_*\|^{1/2}\cdot \|G^{1/2}_*|F_2|^2F G^{1/2}_*\|^{1/2}.
\end{align}
The terms in the r.h.s. above can be estimated with the first inequality of (\ref{b_GKG}) and the last inequality of (\ref{b_GKG1}).
In the last inequality of (\ref{est}) we used that since $F\le 1$ and $\widetilde K_1\widetilde K_1^*\le c\cdot |\widetilde K_1|$,
\[
G^{1/2}_*F\widetilde K_1F\widetilde K_1^* FG^{1/2}_*\le G^{1/2}_*F\widetilde K_1\widetilde K_1^* FG^{1/2}_*\le c\cdot G^{1/2}_*F|\widetilde K_1|FG^{1/2}_*.
\]
The expression $\|G^{1/2}_*F\widetilde K_2FF_1G^{1/2}_*\|$ can be estimated similarly.

Now we are left to show (\ref{b_GKG}) -- (\ref{b_GKG1}). To prove the first inequality of (\ref{b_GKG}), we recall first that for any normal $A$ and $B$
\begin{equation}\label{*}
|B^*AB\|\le\|B^*|A|B\|
\end{equation}
Indeed, for any normal $A$ we have
\[
|(Ax,y)|^2\le (|A|x,x) (|A|y,y).
\]
and so putting $Bx$ and $By$ instead of $x$ and $y$ we get (\ref{*}).

Now (\ref{*}), the first inequality in (\ref{bound_re}), and the bound $F\le 1$ yield
\begin{align*}
 F|\widetilde K_\alpha |F&\le F(1-K_{US})F\le 1-FK_{US}F\\
&\Rightarrow\|G^{1/2}_*F|\widetilde K_\alpha| FG^{1/2}_*\|\le \|G^{1/2}_*(1-FK_{US}F)G^{1/2}_*\|\le C,
\end{align*}
since the spectrum of  $FK_{US}F$ belongs to  $[0,1]$ and 
\begin{equation}\label{est_lam}
\max\limits_{0\le \lambda\le 1}\dfrac{1-\lambda}{1+A/n-\lambda}\le 1.
\end{equation}
Moreover, since  $G$ and $G_*$ commute  we have
\begin{align*}
& \|G^{1/2}(z)G^{-1/2}_*\|^2= \|G(z)G^{-1}_*\|\le\max_{|z|=1+A/n,0\le\lambda\le 1}\frac{1+A/n-\lambda}{|z-\lambda|}\le C,
\end{align*}
which gives the second inequality of (\ref{b_GKG}). 

To prove the first inequality  (\ref{b_GKG1}), take  $n$-independent $B>0$ and introduce the projection 
\[\Pi_n=\mathbf{1}_{|S_{12}|^2\le Bn\log n}.\]
From the definition (\ref{F}) it is evident that for sufficiently big $B$ we can write
\begin{align}\label{in_F}
&\|(1-\Pi_n)FF_\alpha\|\le C\max_{x>B'\log n} xe^{-x}\le C/n^3,\\
& 0\le \dfrac{c_0\alpha_2|S_{12}|^2}{n}(1-F^2)^{-1}\Pi_n= \max_{0\le x\le B'\log n} x(1-e^{-2x})^{-1}\le C(B)\log n\notag\\ \Rightarrow&
|F_\alpha|\Pi_n\le C/n+C(B)\log n(1-F^2)\le C/n+C(B)\log n(1-FK_{US}F)
\notag\end{align}
with $B'=c_0\alpha_2B$.
Using the first inequality above, the bound $\|G_*\|\le Cn$, and the fact that $FK_{US}F$ commute with $G_*$, we get
\begin{align*}
G_*^{1/2}F_\alpha FK_{US}FG_*^{1/2}&=G_*^{1/2}F_\alpha ((1-\Pi_n)+\Pi_n)FK_{US}FG_*^{1/2}
\\&=O(n^{-2})+
G_*^{1/2}F_\alpha \Pi_nG_*^{1/2}FK_{US}F.
\end{align*}
In addition the third line of (\ref{in_F}) and (\ref{est_lam}) yield
\begin{align*}
&G_*^{1/2}|F_\alpha| \Pi_nG_*^{1/2}\le  C+C'\log n\,G_*^{1/2}(1-FK_{US}F)G_*^{1/2}\le C\log n.
\end{align*}
 The proofs of the 
 other inequalities of (\ref{b_GKG1}) are similar to the proof of the first one.

Thus we obtain (\ref{b_G1}).
Since by (\ref{M_hat}) $\hat K=\hat K_0+O(\tilde{\beta}^{-1})$, we have
\[\widehat G=\widehat G_1(1+O(\tilde{\beta}^{-1})\widehat G_1)^{-1}=\widehat G_1\big(1+O(\log^2n(n/\tilde{\beta}))\big).\]
Combined with (\ref{b_G1}) the relation finishes the proof of Lemma \ref{l:bG}.
$\square$

 \section{Proof of Proposition \ref{p:repr}.}\label{s:6}
We start with a detailed study of the operator $H$ of (\ref{H_op}). 
Set
\[ U= U_1U_2^*,\quad S=S_1S_2^{-1}
\]
 and use two simple formulas, valid for any diagonal $2\times 2$ matrices $A$ and $B$,
 \begin{align*}
 &\Tr AUBU^*=\Tr AB-|U_{12}|^2(A_{11}-A_{22})(B_{11}-B_{22}),\\
& \Tr ASBS^{-1}=\Tr AB+|S_{12}|^2(A_{11}-A_{22})(B_{11}-B_{22}).
\notag \end{align*}
Using (\ref{Q}) and changing 
\begin{align}\label{ch_beta}
&\hat\rho\to \tilde\beta^{-1/2}\hat\rho,\quad  \hat\rho'\to \tilde\beta^{-1/2}\hat\rho'\\ \notag
&\hat\tau\to \tilde\beta^{-1/2}\hat\tau, \quad \hat\tau'\to \tilde\beta^{-1/2}\hat\tau',
\end{align}
 in (\ref{H_op}) (note that this gives the Jacobian $\tilde\beta^{2}$), we get
\begin{align*}
H=&\tilde\beta^2 \cdot e^{-\tilde\beta \cdot w}(1-n_{1}n_{2}/\tilde\beta^2)(1-n_{1}'n_{2}'/\tilde\beta^2)\\
&\cdot\exp\Big\{(n_{1}+n_{2}+n_{1}'+n_{2}')d
-\Tr\hat\rho U \hat\tau ' S^{-1}- \Tr\hat\rho' U^* \hat\tau S-(n_{1}+n_{2})(n_{1}'+n_{2}')w/\tilde\beta\Big\}
\\ d=&1-|U_{12}|^2+|S_{12}|^2,\quad w=|U_{12}|^2+|S_{12}|^2.
\end{align*}
Writing
\begin{align*}
&\Tr\hat\rho U\hat\tau 'S^{-1}=(\rho,A\tau'),\quad \Tr\hat\rho'U^*\hat\tau S=(\rho',B\tau)\\
&A_{ij}=U_{ij}S^{-1}_{ji},\quad B_{ij}=U^*_{ij}S_{ji}
\end{align*}
and using that
\begin{align*}
(\rho,A\tau')^2=-2\det A\,\rho_1\rho_2\tau_1'\tau_2',\quad(\rho',B\tau)^2=-2\det B\,\rho_1'\rho_2'\tau_1\tau_2,
\end{align*}
we obtain
\begin{align}\label{H_gr}
H\big|_{\mathcal{P}_6}=&\tilde\beta^2 \cdot \exp\Big\{-\tilde{\beta}\cdot w + (n_{1}+n_{2}+n_{1}'+n_{2}')d-(n_{1}+n_{2})(n_{1}'+n_{2}')w/\tilde{\beta}\Big\}\\ \notag
&\times\big(1+(\rho,A\tau')(\rho',B\tau)+\mdet A\cdot \mdet B\, n_1n_2n_1'n_2'\big)\cdot (1-n_{1}n_{2}/\tilde{\beta}^2)(1-n_{1}'n_{2}'/\tilde{\beta}^2).
\end{align}
Introduce the basis  $e_1=1,\,e_2=n_1,\,e_3=n_2,\, e_4=n_1n_2,\,e_5=\rho_1\tau_2,\,e_6=\rho_2\tau_1$ of $\mathcal{P}_6$.
Denote the space spanned  on the first 4 vectors as $\mathcal{P}_4$ and
represent  $H$ in this basis by the block $6\times 6$ matrix with $H^{(11)}$ corresponding to the projection on $\mathcal{P}_4$. Then using (\ref{H_gr}) we obtain
\begin{align}\label{H^ij}
&H= \left(\begin{array}{cc}H^{(11)}&H^{(12)}\\
H^{(21)}&H^{(22)}\end{array}\right),\quad
H^{(22)}=K_{US}\left(\begin{array}{cc}A_{11}B_{22}&A_{12}B_{12}\\
A_{21}B_{21}&A_{22}B_{11}\end{array}\right)=\left(\begin{array}{cc}h_{11}&h_{12}\\
h_{21}&h_{22}\end{array}\right),\\
&H^{(21)}=\left(\begin{array}{cccc}2x_d& x&x&0\\
-2\overline x_d&-\overline x&-\overline x&0\end{array}\right),\quad H^{(12)}=\left(\begin{array}{cc}0&0\\  y& -\overline{ y}\\
 y& -\overline{ y}\\
2 y_d&-2\overline{ y_d}
\end{array}\right).
\notag\end{align}
Here and below
 $h_{ij}, x,y,x_d,y_d$ are  "difference" operators whose kernels are defined with the functions
\begin{align}\label{hxy}
h_{ij}=&h_{ijU}h_{ijS},\quad h_{ijU}=U_{ij}^2K_U,\quad h_{ijS}=\bar S_{ij}^2K_S\\
x=&x_Ux_S,\quad x_U=U_{11}U_{12}K_U,\quad x_S=\bar S_{11}\bar  S_{12}K_{S},\quad x_d=x\cdot d,\notag\\
 y=&y_Uy_S,\quad y_U=U_{11}\bar  U_{12}K_U\quad y_S=\bar S_{11}S_{12}K_{S},\quad
y_d=y\cdot d,
\notag\end{align}
and $\bar x,\bar y,\bar x_d,\bar y_d$ mean the complex conjugate kernels.
Now let us study the structure of $H^{(11)}$. Using (\ref{H_gr}) and the relations
\[\det A=\det B=d,\quad (A\rho,\tau')(B\rho',\tau)\big|_{\mathcal{P}_4}=-d(n_1n_1'+n_2n_2')+ |U_{12}|^2|S_{12}|^2(n_1+n_2)(n_1'+n_2')\]
we continue to transform $H$ as 
\begin{align*}
H\big|_{\mathcal{P}_4}=K_{US}&\cdot e^{d(n_1+n_2+n_1'+n_2')} \big(1-w(n_1+n_2)(n_1'+n_2')/\tilde{\beta}+2w^2n_1n_2n_1'n_2'/\tilde{\beta}^2\big)\\
&\times\big(1-(n_1n_2+n_1'n_2')/\tilde{\beta}^2+n_1n_2n_1'n_2'/\tilde{\beta}^4\big)\\
&\times \big(1-d(n_1n_1'+n_2n_2')+|U_{12}|^2|S_{12}|^2(n_1+n_2)(n_1'+n_2')+d^2n_1n_2n_1'n_2'\big)\\
=&K_{US}\cdot e^{d(n_1+n_2+n_1'+n_2')}\Big(1-w(n_1+n_2)(n_1'+n_2')/\tilde{\beta}-(n_1n_2+n_1'n_2')/\tilde{\beta}^2\\
&-d(n_1n_1'+n_2n_2')
+|U_{12}|^2|S_{12}|^2(n_1+n_2)(n_1'+n_2')
\\&+(d^2+2w^2/\tilde{\beta}^2+2dw/\tilde{\beta}+1/\tilde{\beta}^4-4w|U_{12}|^2|S_{12}|^2/\tilde{\beta}) n_1n_2n_1'n_2'\Big).
\end{align*}
Represent  $H^{(11)}=K_{US}\cdot K$ and  observe that to find the coefficients of $K$ we can represent $H$ as a polynomials
with respect to $n_1,n_2,n_1',n_2'$ and the coefficients of this polynomials gives the coefficients of $K$. In particular,
\begin{align*}
&K_{11}\sim n_1'n_2',\quad K_{21}\sim n_1n_1'n_2',\quad K_{31}\sim n_2n_1'n_2',\quad K_{41}\sim n_1n_2n_1'n_2',\quad\\
&K_{42}\sim n_1n_2n_2',\quad K_{43}\sim n_1n_2n_1',\quad K_{24}\sim n_1,\quad K_{34}\sim n_2,\quad K_{44}\sim n_1n_2,\quad
\end{align*}
Evidently these and the other coefficient of $K$ can be found as the respective derivatives, taken at the point $(n_1,n_2,n_1',n_2')=(0,0,0,0)$. 

Now we return to the proof of Proposition \ref{p:repr}. In order to transform (\ref{trans1}) to (\ref{repr1}) with an appropriate $\widehat M$ and $\widehat{K}$ satisfying (\ref{M_hat}) -- (\ref{repr2})
 we are going to consider the matrix $K$ after the transformation
\[K_T=TKT,\quad T=\left(\begin{array}{cccc}
0&0&0&\tilde{\beta}\\0&0&1&0\\0&1&0&0\\\tilde{\beta}^{-1}&0&0&0
\end{array}\right)
\]
It is easy to see that 
\begin{align*}
&K_{T12}=\tilde{\beta} K_{43},\,K_{T13}=\tilde{\beta} K_{42},\,K_{T24}=\tilde{\beta} K_{31},\,K_{T34}=\tilde{\beta} K_{21},\,K_{T14}=\tilde{\beta}^2K_{41}.
\end{align*}
All the rest coefficients $K$ change the places  or are multiplied by $1$, $\tilde{\beta}^{-1}$ or even $\tilde{\beta}^{-2}$.
Thus, to obtain representation (\ref{repr1})- (\ref{repr2}), we need to control the elements of $\widehat K$ written above. The following lemma allows
to understand the order of the operators, which will appear in the coefficients of $K$.
\begin{lemma}\label{l:U,V} 
\begin{align}\label{U,V.1}
&K_{US}|U_{12}|^2=\tilde{\beta}^{-1}+O((1-K_{US})\tilde{\beta}^{-1}),\quad K_{US}|S_{12}|^2=\tilde{\beta}^{-1}+O((1-K_{US})\tilde{\beta}^{-1}),\\
&K_{US}|U_{12}|^4=2\tilde{\beta}^{-2}+O((1-K_{US})\tilde{\beta}^{-2}),\quad K_{US}|S_{12}|^4=2\tilde{\beta}^{-2}+O((1-K_{US})\tilde{\beta}^{-2}),\notag\\
& K_{US}|U_{12}|^2|S_{12}|^2=\tilde{\beta}^{-2}+O((1-K_{US})\tilde{\beta}^{-2}).\notag
\end{align}
 We recall  that all operators here are self adjoint and commute with each other,
hence the relations mean the ones for the corresponding eigenvalues.
\end{lemma}
The proof or the lemma will be given at the end of  the proof of Lemma \ref{l:R} (see  the argument above (\ref{mu^l})).

Coming back to the coefficients of $K$, compute first
\begin{align*}
K_{US}\cdot K_{41}&=K_{US}\cdot\frac{\partial^4K}{\partial n_1\partial n_2\partial n_1'\partial n_2'}\Big|_{(0,0,0,0)}=K_{US}\big(d^4-2d^3+d^2+4d^2|U_{12}|^2|S_{12}|^2\\
&-4d^2w\tilde{\beta}^{-1}-2d^2\tilde{\beta}^{-2}+2dw\tilde{\beta}^{-1}+2w^2\tilde{\beta}^{-2}-4w|U_{12}|^2|S_{12}|^2/\tilde{\beta}+\tilde{\beta}^{-4}\big)\\
&=K_{US}\Big(d^2w^2-2d^2\tilde{\beta}^{-2}-2 dw\tilde{\beta}^{-1}\Big)+O(\tilde{\beta}^{-3})\\
&=\tilde{\beta}^{-2}\widetilde K',\quad \widetilde K'=O(1-K_{US}).
\end{align*}
Here we have used the relation (which follows from the definition of $d$ and $w$)
\[
d^4-2d^3+d^2+4d^2|U_{12}|^2|S_{12}|^2=d^2w^2,\quad 4d^2 w-4dw=4dw (|S_{12}|^2-|U_{12}|^2),
\]
and the  lemma above.

Similarly
\begin{align*}
&K_{US}\cdot K_{21}=K_{US}\cdot K_{31}=K_{US}\cdot K_{42}=K_{US}\cdot K_{43}\\
&=K_{US}(d^3-d^2-2dw/\tilde{\beta}-d/\tilde{\beta}^2+2d|U_{12}|^2|S_{12}|^2)=\tilde{\beta}^{-1}\widetilde K\\ 
\end{align*}
with
\[
\widetilde K=O(1-K_{US}).
\]
In addition,
\begin{align*}
&K_{ii}=1+O(\tilde{\beta}^{-1}),\quad  i=1,\dots 4,\\
&K_{ij}=O(\tilde{\beta}^{-1}), \quad ( i,j)=(2,3)\,\,or\,\,(3,2).
\end{align*}
Observe now that the operator $\tilde{\mathcal{F}}=\mathcal{F}\big|_{\mathcal{P}_6}$ in (\ref{trans1}) after the change (\ref{ch_beta}) in our basis have the block 
diagonal form, where
a $4\times 4$ upper left block has the form $T\widehat FT$, where $\widehat F$ is given by (\ref{repr2}), and a $2\times 2$ bottom left block is $I$.
In addition, $\tilde f$ and $\tilde g$ are  spanned on $e_1,e_2,e_3,e_4$ and after the change (\ref{ch_beta}) their restriction on $\mathcal{P}_4$ have the form
$\tilde f=\beta^{-1}T\hat f$, $\tilde g=\beta^{-1}T\hat g$. Thus  we are interested in
the upper left block $G^{(11)}$ of the resolvent $G=(\tilde {\mathcal F} H\tilde {\mathcal{F}}-z)^{-1}$, and so (\ref{trans1}) yields
\begin{align*}
\mathcal{R}_{n\beta}^{+-}(E,\varepsilon,\xi)=\frac{C_{E,\varepsilon}e^{c_0(\alpha_1-\alpha_2)}}{2\pi i}\oint_{\omega_A}z^{n-1}(T\widehat G^{(11)}(z)T\widehat f,\widehat g)dz.
\end{align*}
But by the Schur compliment formula
\[
TG^{(11)}(z)T=\Big(\widehat FT\big(H^{(11)}- H^{(12)}(H^{(22)}-z)^{-1}H^{(21)}\big)T\widehat F-z\Big)^{-1},
\]
and so we are left to prove that 
\begin{equation}\label{M_T}
\widehat M= \widehat FT(H^{(11)}- H^{(12)}(H^{(22)}-z)^{-1}H^{(21)})T\widehat F
\end{equation}
satisfies (\ref{M_hat}) -- (\ref{repr2}).

According to the consideration above, $TH^{(11)}T$ has the form (\ref{M_hat}) -- (\ref{repr2}).  The estimate on $H^{(12)}(H^{(22)}-z)^{-1}H^{(21)}$ is given in the following lemma
\begin{lemma}\label{l:R}
Set  $G^{(2)}(z):=(H^{(22)}-z)^{-1}$.  Then
for any $z:|z|=1+A/n$ the operator $H^{(12)}G^{(2)}(z)H^{(21)}$ has the form
\begin{align}\notag
&H^{(12)}G^{(2)}H^{(21)}=
\left(\begin{array}{cccc}0&0&0&0\\
2R_{1d}&R&R&0\\
2R_{1d}&R&R&0\\
4R_{dd}&2R_{d1}&2R_{d1}&0
\end{array}\right), 
\end{align}
where
\begin{align}
R=& yG^{(2)}_{11} x+
 {\overline y}G^{(2)}_{22}{\overline x}
-yG^{(2)}_{12}{\overline x}
-{\overline y}G^{(2)}_{21} x,
\label{R}\end{align}
 $R_{1d}$ can be obtained from $R$, if we replace  $x$ with $x_d$, to obtain
$R_{d1}$, one should replace $y$ with $y_d$,  to obtain $R_{dd}$, one should replace $x$, $y$ with $x_d$
 $y_d$, and the operators $x, y, x_d, y_d$ are the same as in (\ref{H^ij}). 

The operators $R,R_{1d},R_{d1},R_{dd}$ are normal and satisfy the bound
\begin{align}\label{lR.1}
|R|+|R_{1d}|+|R_{d1}|+|R_{dd}|\le C\tilde{\beta}^{-2}(1-K_{US})+O(\tilde\beta^{-3}),\quad 1-K_{US}\le C(\widetilde\Delta_U+\widetilde\Delta_S)/\tilde\beta
\end{align}
\end{lemma}
The lemma gives that (\ref{M_T}) indeed satisfies (\ref{M_hat}) -- (\ref{repr2}), and (\ref{bound_re}),
which finishes the proof of Proposition \ref{p:repr}.

$\square$

\textit{Proof of Lemma \ref{l:R}.}
Let us prove  (\ref{lR.1}) for $R$ of (\ref{R}). For $R_{1d}$, $R_{d1}$, $R_{dd}$ the proof is the same.
To simplify notations set 
\[
H^{(22)}=h=\hat h+\tilde h,
\]
where $\hat h$ is the diagonal part of $H^{(22)}$, and $\tilde h$ is its off diagonal part, and denote
\[G^{(2)}_0:=(\hat h-z)^{-1}.\]
It is easy to see that
\begin{align}\label{tr_b}
\|h_{12}\|\le\int|U_{12}|^2|S_{12}|^2K_UK_SdUdS\le\tilde{\beta}^{-2},\quad \|x\|\le \tilde{\beta}^{-1},\quad \|y\|\le \tilde{\beta}^{-1},
\end{align}
(recall that by (\ref{H^ij}) $h_{ij}=U_{ij}^2\bar S_{ij}^2K_UK_S$).
Hence, writing
\begin{align*}
G^{(2)}=G^{(2)}_0-G^{(2)}_0\tilde hG^{(2)}_0+r,\quad r:=G^{(2)}_0\tilde hG^{(2)}_0\tilde hG^{(2)},
\end{align*}
and using the bounds above combined with (\ref{H^ij}), we get
\[
\|r\|\le Cn^{3}\tilde{\beta}^{-4}\quad \Rightarrow\quad \|H^{(12)}rH^{(21)}\|\le n^3\tilde{\beta}^{-6}<\tilde{\beta}^{-3}.
\]
Consider $\widehat R$ which has the same form as (\ref{R}) but with $G^{(2)}$ replaced by $G^{(2)}_0$ . Then the second two terms become zeros
and
\begin{align*}
\widehat R=y(h_{11}-z)^{-1}x+\bar y(h_{22}-z)^{-1}\bar x=\widehat R_1+\widehat R_2.
\end{align*}
Let us study the operator
\[\widehat R_1=\sum_{p=0}^\infty \frac{y(h_{11})^px}{z^p}=\sum_{p=0}^\infty \frac{(y_U(h_{11U})^px_U)\otimes (y_S(h_{11S})^px_S)}{z^p},
\]
where $y_U,h_{11U},x_U$ (see (\ref{hxy})) are  integral operators on $L_2(U)$ with the "difference" kernels
of the form $v(U_1U^{-1}_2)$, and $y_S,h_{11S},x_S$  are the "difference" integral operators on $L_2(S)$.  
Here  $L_2(U)$ and $L_2(S)$ denote the subspaces of  even functions $\varphi(U)=\varphi(-U)$ (or $\varphi(S)=\varphi(-S)$). Since our operators
preserve the evenness, it suffices to study only these subspaces. It is known that
\[L_2(U)=\displaystyle\oplus_{l=0}^\infty L^{(l)U}, \quad L^{(l)U}=\mathrm{Lin}\,\{t^{(l)U}_{mk}\}_{m,k=-l}^l\]
where $\{t^{(l)U}_{mk}(U)\}_{m,k=-l}^l$ are the coefficients of the irreducible representation of the shift operator $T_U\widetilde U=U\widetilde U$. 
It follows from the properties of the unitary representation that
\begin{align*}
t^{(l)U}_{mk}(U^{-1})=\overline{t^{(l)U}_{km}(U)},\quad t^{(l)U}_{mk}(U_1U_2)=\sum t^{(l)U}_{mj}(U_1)t^{(l)U}_{jk}(U_2).
\end{align*}
According to \cite{Vil:68}, Chapter III,
\[t^{(l)U}_{mk}(U)=e^{-i(m\phi+k\psi)/2}P^{(l)}_{mk}(\theta),\]
where
\begin{align}\notag
 &P^{(l)}_{mk}(\cos\theta)=\frac{c_{mk}}{2\pi}\int_0^{2\pi}d\varphi(\cos(\theta/2)+i\sin(\theta/2)e^{i\varphi})^{l+k}
(\cos(\theta/2)+i\sin(\theta/2)e^{-i\varphi})^{l-k}e^{i(m-k)\varphi},\\
&c_{mk}=\Big(\frac{(l-m)!(l+m)!}{(l-k)!(l+k)!}\Big)^{1/2},\quad U=\left(\begin{array}{rr}\cos(\theta/2)e^{i(\phi+\psi)/2}&i\sin(\theta/2)e^{i(\phi-\psi)/2}\\
i\sin(\theta/2)e^{-i(\phi-\psi)/2}&\cos(\theta/2)e^{-i\phi+\psi)/2}\end{array}\right).
\label{t_mk}\end{align}
In addition (see \cite{Vil:68}, Chapter III),
\begin{align}\label{P(1)}
P^{(l)}_{mm}(1-x)=1-x(l+m)(l+m+1)/2+O(x^2).
\end{align}
It is known also
that $\{t^{(l)U}_{mk}(U)\}_{m,k=-l}^l$  make an orthonormal basis in  $L^{(l)U}$.

For any function $v(U)$ consider the matrix $ v^{(l)U}=\{v^{(l)U}_{mk}\}$ defined as
\[ v^{(l)U}_{mk}:=\int v(U)\overline{ t^{(l)U}_{mk}(U)}dU.
\]
It is easy to see that if we consider an integral operator $\widehat v$ with the kernel $v(U_1U^{-1}_2)$, then for any $\varphi(U)\in L^{(l)U}$
\begin{align*}
(\widehat v\varphi)(U)=\int v(UU_1^{-1})\sum_{mk} \varphi_{mk} t^{(l)U}_{mk}(U_1)dU_1=
\int v(\widetilde U)\sum_{mk} \varphi_{mk} t^{(l)U}_{mk}(\widetilde U^{-1}U)d\widetilde U\\=
\int v(\widetilde U)\sum_{mkj} \varphi_{mk} t^{(l)U}_{mj}(\widetilde U^{-1}) t^{(l)U}_{jk}( U)d\widetilde U
=\sum v^{(l)U}_{jm}\varphi_{mk}t^{(l)U}_{jk}( U).
\end{align*}
Hence, denoting $\Pi_l$ the orthogonal projection on $L^{(l)U}$, one can see that $L^{(l)U}$ reduces $\widehat v$ and $\widehat v^{(l)U}=\Pi_l\widehat v \Pi_l$ is uniquely  defined by the matrix $v^{(l)U}$.  Moreover, for any functions $v_1$ and $v_2$ it is evident
that $\widehat v_1\widehat v_2$ is also a "difference" operator, hence it commutes with $\Pi_l$, and if the matrices $v_1^{(l)U}$ and $v_2^{(l)U}$ correspond
to $ \widehat v_1$ and $\widehat v_2$, then
\[(\widehat v_1\widehat v_2)^{(l)U}= v_1^{(l)U}v_2^{(l)U}.
\]
Let us find the matrices, corresponding to  $h_{11U},y_U,x_U$ (see (\ref{hxy}))  in $L^{(l)U}$. Using  (\ref{t_mk}) it is easy to see that
\begin{align*}
(x_U^{(l)U})_{km}=&\int U_{11}U_{12}K_U(|U_{12}|^2)\overline{ t^{(l)U}_{mk}(U)}dU
=\tilde\beta\int_0^\pi\sin\theta d\theta \int_0^{2\pi}d\phi d\psi\sin(\theta/2)\cos(\theta/2)\\&\times e^{-\tilde\beta\sin^2(\theta/2))}e^{i\phi}
e^{im\phi+in\psi}\overline{ P^{(l)}_{-1,0}(\cos\theta)}=\delta_{m,-1}\delta_{n,0}\lambda^{(l)U}_{-1,0},
\end{align*}
where we set
\begin{align}\label{la_-1,0}
\lambda^{(l)U}_{-1,0}=\frac{\tilde\beta}{2}\int_0^\pi  e^{-\tilde\beta\sin^2(\theta/2)}
P^{(l)}_{-1,0}(\cos\theta)\sin^2\theta d\theta.
\end{align}
Hence, denoting $E_{ij}$ the matrix which has only $ij$th  entry equal to 1, and all other entries equal to 0, we get
\[x_U^{(l)U}=E_{-1,0}\lambda^{(l)U}_{-1,0}.\]
Introduce also the eigenvalues $\lambda^{(l)U}$ of $K_U$. Repeating the argument above, we have
\begin{align*}
K_U^{(l)U}=E_{0,0}\lambda^{(l)U},
\end{align*}
where, using (\ref{P(1)}), we obtain
\begin{align}\label{lambda_l}
\lambda^{(l)U}=&\tilde\beta\int_0^\pi  e^{-\tilde\beta\sin^2(\theta/2)}
P^{(l)}_{0,0}(\cos\theta)\sin\theta d\theta=\tilde{\beta}^{-1}\int_0^1e^{-\tilde{\beta} x}P^{(l)}_{00}(1-2x)dx\\
=&1-{l(l+1)}/{\tilde{\beta}}+O(l^4/\tilde{\beta}^2).
\notag\end{align}
To find an asymptotic behaviour of $\lambda^{(l)U}_{-1,0}$, observe that formulas  (\ref{t_mk}) and (\ref{P(1)} )  yield
\begin{align*}
&2i(1+1/l)^{1/2}\sin(\theta/2)\cos(\theta/2)P^{(l)}_{-1,0}(\cos\theta)\\
=&\int_0^{2\pi}\frac{d\varphi}{2\pi}(\cos(\theta/2)+i\sin(\theta/2)e^{i\varphi})^{l}
\cos(\theta/2)+i\sin(\theta/2)e^{-i\varphi})^{l}2i\cos\varphi\sin(\theta/2)\cos(\theta/2)\\=
&\int_0^{2\pi}\frac{d\varphi}{2\pi}(\cos(\theta/2)+i\sin(\theta/2)e^{i\varphi})^{l}
(\cos(\theta/2)+i\sin(\theta/2)e^{-i\varphi})^{l}\\
&\times\Big((\cos(\theta/2)+i\sin(\theta/2)e^{i\varphi})
(\cos(\theta/2)+i\sin(\theta/2)e^{-i\varphi})-\cos^2(\theta/2)+\sin^2(\theta/2)\Big)\\
=&P^{(l+1)}_{0,0}(\cos\theta)-P^{(l)}_{0,0}(\cos\theta)(1-2\sin^2(\theta/2)).
\end{align*}
Hence
\begin{align*}
(1+1/l)^{1/2}\lambda^{(l)U}_{-1,0}=&(\lambda^{(l+1)U}-\lambda^{(l)U})/2+O(\tilde{\beta}^{-1})=-(l+1)/\tilde{\beta}+O(l^2\tilde{\beta}^{-2})+O(\tilde{\beta}^{-1})\\
\Rightarrow& |\lambda^{(l)U}_{-1,0}|^2\le C_0(1-\lambda^{(l)U})/\tilde{\beta}.
\end{align*}
Similarly
\begin{align*}
&y_U^{(l)U}=E_{0,-1}\overline{\lambda^{(l)U}_{-1,0}},\quad h_{11U}^{(l)U}=E_{-1,-1}\lambda^{(l)U}_{-1,-1},\end{align*}
where we set
\begin{align}\label{la_-1,-1}
& \lambda^{(l)U}_{-1,-1}=\tilde\beta\int_0^\pi  e^{-\tilde\beta\sin^2(\theta/2)}\cos^2(\theta/2)P^l_{-1,-1}(\cos\theta)\sin\theta d\theta=\lambda^{(l)U}+O(\tilde{\beta}^{-1})
\end{align}
and in the last relation we used (\ref{P(1)}). Thus, for any $p$
\begin{align}\label{xy_U}
(y_U(h_{11U})^px_U)^{(l)U}=|\lambda^{(l)U}_{-1,0}|^2(\lambda^{(l)U}_{-1,-1})^pE_{00}.
\end{align}

The analysis of $(y_S(h_{11S})^px_S)$ is very similar, the difference is that for the hyperbolic group the irreducible representations
are labelled by the continuous parameter $l'=-\frac{1}{2}+i\rho$,  $\rho\in\mathbb{R}$,
\[
t^{(l')S}_{mk}=e^{i(m\phi+k\psi)}\mathcal{B}^{(l)}_{mk}(\theta),\quad m,k\in\mathbb{Z},
\]
and $\mathcal{B}^{(l)}_{mk}(\theta)$ has the form (\ref{t_mk}) with $\cos(\theta/2)$ replaced by $\cosh(\theta/2)$,
$i\sin(\theta/2)$ replaced by $\sinh(\theta/2)$ and $c_{mk}$ replaced by 1 (see \cite{Vil:68}, Chapter VI) .
Then the same argument yields that
\begin{align}\label{xy_V}
&(y_S(h_{11V})^px_S)^{(l')}=|\lambda^{(l)S}_{-1,0}|^2(\lambda^{(l')S}_{-1,-1})^pE_{00}\\
&|\lambda^{(l')S}_{-1,0}|^2\le C_0(1-\lambda^{(l')S})/\tilde{\beta},\quad \lambda^{(l')S}_{-1,-1}=\lambda^{(l')S}+O(\tilde{\beta}^{-1}),
\notag\end{align}
where $\lambda^{(l)S}_{-1,0}$, $\lambda^{(l')S}$, and $\lambda^{(l')S}_{-1,-1}$, are defined similarly to (\ref{la_-1,0}), (\ref{lambda_l}), and (\ref{la_-1,-1}).
Here the bound for  $|O(\tilde{\beta}^{-1})|<C_0\tilde{\beta}^{-1}$ is uniform in $l$.

This relation combined with (\ref{xy_U})  yields that  $\widehat R_1:L^{(l)U}\otimes L^{(l')S}\to L^{(l)U}\otimes L^{(l')S}$  and the 
only non zero eigenvalue of $\widehat R_1$ in this subspace 
has the form
\[ \lambda^{(ll')}=|\lambda^{(l)U}_{-1,0}|^2|\lambda^{(l')S}_{-1,0}|^2(z- \lambda^{(l)U}_{-1,-1} \lambda^{(l')S}_{-1,-1})^{-1}\]
The bounds (\ref{xy_U}) and (\ref{xy_V}) yield for $|z|>1+2C_0\tilde{\beta}^{-1}$
\[|\lambda^{(ll')}|\le C\tilde{\beta}^{-2}\frac{(1-\lambda^{(l)U})\cdot(1-\lambda^{(l')S})}{|z|-\lambda^{(l)U}\lambda^{(l')S}+O(\tilde{\beta}^{-1})}\le
C\tilde{\beta}^{-2}|1-\lambda^{(l)U}\lambda^{(l')S}|.
\]
Here we have used that for any $0<a,b<1$
\[ab\le a^2+b^2-ab<a^2+b^2-a^2b^2,\]
hence, taking $a^2=1-\lambda^{(l)U}$, $b^2=1-\lambda^{(l')S}$, we obtain the last inequality for $|\lambda^{(ll')}|$. 

Note that (\ref{lambda_l}) and a similar relation for $\lambda^{(l')S}$ combined with the facts that
\[\widetilde\Delta_UL^{(l)U}=l(l+1)L^{(l)U},\quad \widetilde\Delta_SL^{(l')S}=-l'(l'+1)L^{(l')S}\]
prove the second inequality in (\ref{lR.1}).

Assertions of Lemma \ref{l:U,V} can be obtained from the fact that the operators in the l.h.s. of (\ref{U,V.1}) are tensor products of the "difference"
operators on $L_2(U)$ and $L_2(S)$. Hence they are reduced by $L^{(l)U}\otimes L^{(l')S}$, and since the kernels depend on $|U_{12}|^2$ and
$|S_{12}|^2$, the corresponding matrices have the form $\mu^{(l)}\nu^{(l')}E_{00}\otimes E_{00}$, where $\mu^{(l)}, \nu^{(l')}$ -are corresponding eigenvalues.
For example, for the first operator in (\ref{U,V.1}) $\nu^{(l')}=\lambda^{(l')S}$ and
\begin{align}\label{mu^l}
\mu^{(l)}=&\tilde\beta\int_0^\pi \sin^2(\theta/2) e^{-\tilde\beta\sin^2(\theta/2)}
P^{(l)}_{0,0}(\cos\theta)\sin\theta d\theta=\tilde{\beta}^{-1}\int_0^1xe^{-\tilde{\beta} x}P^{(l)}_{00}(1-2x)dx\\
=&1/\tilde{\beta}-2{l(l+1)}/\tilde{\beta}^2+O(l^4/\tilde{\beta}^3).
\notag\end{align}
The first relation of (\ref{U,V.1}) follows from the above one combined with the analogue of (\ref{lambda_l}) for $\lambda^{(l')S}$. The other relations of (\ref{U,V.1}) can be obtained similarly.

To complete the proof of the lemma we are left to consider the part of $R$ which  can be obtained if we replace $G^{(2)}$ with
$G^{(2)}_0\tilde hG^{(2)}_0$. For this replacement the first two terms of (\ref{R}) are zero.  Set
\[R_3=y(h_{11}-z)^{-1}h_{12}(h_{22}-z)^{-1}\bar x.
\]
Repeating the above argument  we obtain that $R_3:L^{(l)U}\otimes L^{(l')S}\to L^{(l)U}\otimes L^{(l')S}$  and the 
only non zero eigenvalue $\widetilde\lambda^{(ll')}$ of $R_3$ in this subspace 
has the form
\[\widetilde\lambda^{(ll')}=\frac{|\lambda^{(l)U}_{-1,0}|^2|\lambda^{(l')S}_{-1,0}|^2\lambda^{(l)U}_{-1,1}\lambda^{(l')S}_{-1,1}}{(z- \lambda^{(l)U}_{-1,-1} \lambda^{(l')S}_{-1,-1})^2},
\]
where $\lambda^{(l)U}_{-1,1}$ and $\lambda^{(l')S}_{-1,1}$ by (\ref{tr_b}) satisfy the trivial bound
\[|\lambda^{(l)U}_{-1,1}\lambda^{(l)S}_{-1,1}|\le\|h_{12}\|\le \tilde{\beta}^{-2}.
\]
The bound, (\ref{xy_U}) and (\ref{xy_V}) yield
\[ |\widetilde\lambda^{(ll')}|\le C\tilde{\beta}^{-4}\frac{|1-\lambda^{(l)U}|\cdot|1-\lambda^{(l)V}|}{||z|-\lambda^{(l)U}\lambda^{(l')V}+O(\tilde{\beta}^{-1})|^2}\le Cn/\tilde{\beta}^4<\tilde{\beta}^{-3}.
\]
The same bound is valid for
\[R_4=\bar y(h_{22}-z)^{-1}h_{21}(h_{11}-z)^{-1}x.
\]
These bounds  complete the proof of the lemma for $R$. For $R_{1d},R_{d1},R_{dd}$ the proof is the same.

$\square$
\section{Appendix}

\subsection{Grassmann integration}
Let us consider two sets of formal variables
$\{\psi_j\}_{j=1}^n,\{\overline{\psi}_j\}_{j=1}^n$, which satisfy the anticommutation
conditions
\begin{equation}\label{anticom}
\psi_j\psi_k+\psi_k\psi_j=\overline{\psi}_j\psi_k+\psi_k\overline{\psi}_j=\overline{\psi}_j\overline{\psi}_k+
\overline{\psi}_k\overline{\psi}_j=0,\quad j,k=1,\ldots,n.
\end{equation}
Note that this definition implies $\psi_j^2=\overline{\psi}_j^2=0$.
These two sets of variables $\{\psi_j\}_{j=1}^n$ and $\{\overline{\psi}_j\}_{j=1}^n$ generate the Grassmann
algebra $\mathfrak{A}$. Taking into account that $\psi_j^2=0$, we have that all elements of $\mathfrak{A}$
are polynomials of $\{\psi_j\}_{j=1}^n$ and $\{\overline{\psi}_j\}_{j=1}^n$ of degree at most one
in each variable. We can also define functions of
the Grassmann variables. Let $\chi$ be an element of $\mathfrak{A}$, i.e.
\begin{equation}\label{chi}
\chi=a+\sum\limits_{j=1}^n (a_j\psi_j+ b_j\overline{\psi}_j)+\sum\limits_{j\ne k}
(a_{j,k}\psi_j\psi_k+
b_{j,k}\psi_j\overline{\psi}_k+
c_{j,k}\overline{\psi}_j\overline{\psi}_k)+\ldots.
\end{equation}
For any
sufficiently smooth function $f$ we define by $f(\chi)$ the element of $\mathfrak{A}$ obtained by substituting $\chi-a$
in the Taylor series of $f$ at the point $a$. Since $\chi$ is a polynomial of $\{\psi_j\}_{j=1}^n$,
$\{\overline{\psi}_j\}_{j=1}^n$ of the form (\ref{chi}), according to (\ref{anticom}) there exists such
$l$ that $(\chi-a)^l=0$, and hence the series terminates after a finite number of terms and so $f(\chi)\in \mathfrak{A}$.

Following Berezin \cite{Ber}, we define the operation of
integration with respect to the anticommuting variables in a formal
way:
\begin{equation*}
\intd d\,\psi_j=\intd d\,\overline{\psi}_j=0,\quad \intd
\psi_jd\,\psi_j=\intd \overline{\psi}_jd\,\overline{\psi}_j=1,
\end{equation*}
and then extend the definition to the general element of $\mathfrak{A}$ by
the linearity. A multiple integral is defined to be a repeated
integral. Assume also that the ``differentials'' $d\,\psi_j$ and
$d\,\overline{\psi}_k$ anticommute with each other and with the
variables $\psi_j$ and $\overline{\psi}_k$. Thus, according to the definition, if
$$
f(\psi_1,\ldots,\psi_k)=p_0+\sum\limits_{j_1=1}^k
p_{j_1}\psi_{j_1}+\sum\limits_{j_1<j_2}p_{j_1,j_2}\psi_{j_1}\psi_{j_2}+
\ldots+p_{1,2,\ldots,k}\psi_1\ldots\psi_k,
$$
then
\begin{equation*}
\intd f(\psi_1,\ldots,\psi_k)d\,\psi_k\ldots d\,\psi_1=p_{1,2,\ldots,k}.
\end{equation*}

   Let $A$ be an ordinary Hermitian matrix with positive real part. The following Gaussian
integral is well-known
\begin{equation}\label{G_C}
\intd \exp\Big\{-\sum\limits_{j,k=1}^nA_{jk}z_j\overline{z}_k\Big\} \prod\limits_{j=1}^n\dfrac{d\,\Re
z_jd\,\Im z_j}{\pi}=\dfrac{1}{\mdet A}.
\end{equation}
One of the important formulas of the Grassmann variables theory is the analog of this formula for the
Grassmann algebra (see \cite{Ber}):
\begin{equation}\label{G_Gr}
\int \exp\Big\{-\sum\limits_{j,k=1}^nA_{jk}\overline{\psi}_j\psi_k\Big\}
\prod\limits_{j=1}^nd\,\overline{\psi}_jd\,\psi_j=\mdet A,
\end{equation}
where $A$ now is any $n\times n$ matrix.

%
%
We will also need the following bosonization formula
\begin{proposition}({\bf see \cite{F:02} })\label{p:supboz}\\
Let $F$ be some function that depends only on combinations 
\begin{align*}
\bar{\phi}\phi:=\Big\{\sum\limits_{\alpha=1}^W \bar{\phi}_{l\alpha}\phi_{s\alpha}\Big\}_{l,s=1}^2,
\end{align*}
and set
\[
d\Phi=\prod\limits_{l=1}^2\prod\limits_{\alpha=1}^W d\Re \phi_{l\alpha} d\Im \phi_{l\alpha}.
\]
Assume also that $W\ge 2$. Then
\begin{equation*}
\int F\left(\bar{\phi}\phi \right)d\Phi=\dfrac{\pi^{2W-1}}{(W-1)!(W-2)!}\int F(B)\cdot \mdet^{W-2} B \,dB,
\end{equation*}
where
$B$ is a $2\times 2$ positive Hermitian matrix, and
\begin{align*}
dB&=\mathbf{1}_{B>0}dB_{11}dB_{22}d\Re B_{12} d\Im B_{12}.
\end{align*}
\end{proposition}

\end{document}